\documentclass[journal]{IEEEtran}

\usepackage{graphicx} 
\usepackage{xcolor} 
\usepackage{soul} 
\usepackage{amsthm}
\usepackage{amsmath}
\newtheorem{definition}{Definition}
\usepackage{booktabs}
\usepackage{multirow}
\usepackage[linesnumbered,ruled,vlined]{algorithm2e}
\usepackage{subfig}
\usepackage{caption}
\usepackage{adjustbox}
\usepackage{xurl}

\usepackage{etoolbox}
\newtoggle{ONEFILE}
\toggletrue{ONEFILE} 

\newtheorem{theorem}{Theorem}
\newtheorem{lemma}{Lemma}
\iftoggle{ONEFILE}{
}{
	\usepackage{xr}
	\externaldocument{supplemental_R1}
}

\begin{document}
%
\title{Harnessing the Power of Ego Network Layers for Link Prediction in Online Social Networks}
%
%
%

\author{Mustafa~Toprak,
        Chiara~Boldrini,
        Andrea~Passarella,
        and~Marco~Conti
\thanks{M. Toprak, C. Boldrini, A. Passarella and M. Conti are with IIT-CNR, Via G. Moruzzi 1, 56124, Pisa, ITALY e-mail: {first.last}@iit.cnr.it}
\thanks{This work was partially funded by the SoBigData++, HumaneAI-Net, MARVEL, and OK-INSAID projects. The SoBigData++, HumaneAI-Net, and MARVEL projects have received funding from the European Union's Horizon 2020 research and innovation programme under grant agreements No 871042, No 952026, No 957337, respectively. The OK-INSAID project has received funding from the Italian PON-MISE program under grant agreement ARS01 00917. This work was also supported by the CHIST-ERA grant CHIST-ERA-19-XAI-010, funded by by MUR (grant No. not yet available), FWF (grant No. I 5205), EPSRC (grant No. EP/V055712/1), NCN (grant No. 2020/02/Y/ST6/00064), ETAg (grant No. SLTAT21096), BNSF (grant No. KP-06-DOO2/5).}
\iftoggle{ONEFILE}{\thanks{\copyright 2022 IEEE. Personal use of this material is permitted.  Permission from IEEE must be obtained for all other uses, in any current or future media, including reprinting/republishing this material for advertising or promotional purposes, creating new collective works, for resale or redistribution to servers or lists, or reuse of any copyrighted component of this work in other works.}}{}
}

\maketitle

\begin{abstract}

Being able to recommend links between users in online social networks is important for users to connect with like-minded individuals as well as for the platforms themselves and third parties leveraging social media information to grow their business. 
Predictions are typically based on unsupervised or supervised learning, often leveraging simple yet effective graph topological information, such as the number of common neighbors. However, we argue that richer information about personal social structure of individuals might lead to better predictions. In this paper, we propose to leverage well-established social cognitive theories to improve link prediction performance. According to these theories, individuals arrange their social relationships along, on average, five concentric circles of decreasing intimacy. We postulate that relationships in different circles have different importance in predicting new links. In order to validate this claim, we focus on popular feature-extraction prediction algorithms (both unsupervised and supervised) and we extend them to include social-circles awareness. We validate the prediction performance of these circle-aware algorithms against several benchmarks (including their baseline versions as well as node-embedding- and GNN-based link prediction), leveraging two Twitter datasets comprising a community of video gamers and generic users. We show that social-awareness generally provides significant improvements in the prediction performance, beating also state-of-the-art solutions like \emph{node2vec} and SEAL, and without increasing the computational complexity. 
Finally, we show that social-awareness can be used in place of using a classifier (which may be costly or impractical) for targeting a specific category of users. 
 \end{abstract}

\begin{IEEEkeywords}
link prediction, social circles, Dunbar's model, Twitter
\end{IEEEkeywords}

%
\IEEEpeerreviewmaketitle

\section{Introduction}
\label{sec:intro}

\IEEEPARstart{I}{n} early 2021, the number of social media users worldwide has reached 4.2 billion, up 13\% from the previous year~\cite{social2021global-state}. Social media are now large ecosystems, where we gather to interact with people we already know or to join communities of people that share our interests. 
%
Clearly, being able to pair together people with similar interests is crucial for social-driven and content-driven platforms like Facebook, Twitter, Instagram, Youtube, etc. This spurred intense research on link prediction, whereby algorithms suggest potential new links with users that are \emph{similar} to each other. Solving this problem is equivalent to finding a needle in a haystack: online social networks (OSN) feature millions of users, the number of potential links between them grows with the square of the number of users, but the actually existing links are quite sparse.  

%

Link prediction algorithms can be broadly divided into \emph{feature extraction} methods and \emph{feature learning} methods~\cite{mutlu2020review}. The former are based on predefined similarity metrics/heuristics, typically computed using topological information about the network, be it local information (such as the number of common neighbors) or global one (such as the length of paths connecting two nodes). Despite their apparent simplicity, local similarity-based prediction algorithms perform generally well, even when compared to more complex approaches~\cite{Martinez2016}, and are also computationally efficient.
The latter are new link prediction methods (leveraging graph neural networks~\cite{kipf2016semi,Zhang2018,cai2020line,chen2019lstm,zhang2021link} and network embeddings~\cite{grover2016node2vec,de2018combining}) that \emph{automatically} learn latent features (i.e., not previously hand-engineered) that reflect the relevant structure of the graph, and, for this reason, are denoted as feature learning methods. These new methods have been shown to perform extremely well on the link prediction task~\cite{grover2016node2vec,Zhang2018}.

One of the well-known mechanisms that drive link formation is the \emph{homophily} principle, which states that individuals tend to bond with those that are more similar to them~\cite{mcpherson2001birds}. Common-neighbors-based algorithms leverage this principle for predicting new links. However, findings from anthropology show that not all bonds are equal: some relationships are just acquaintances and do not imply a significant cognitive engagement. With respect to the meaningful relationships, each individual organises them into concentric \emph{social~circles} where the intimacy progressively decreases from the innermost to the outermost circle~\cite{hill2003social,Zhou2005}. This model of how humans organise their relationships stems from the \emph{social brain hypothesis}, which links our social life to cognitive constraints related to the size of our neocortex~\cite{dunbar1998social}. This hierarchical social structure is known to impact significantly on who we trust~\cite{sutcliffe2015modelling}, on the way information spreads in online social networks~\cite{arnaboldi2014information}, on the diversity of information that can be acquired by users~\cite{aral2011diversity}, and on academic performance~\cite{Paraskevopoulos2021}. This layered social structure is typically represented using an ego network graph~\cite{everett2005ego,lin2001social,mccarty2002structure,hill2003social}, in which the individual, referred to as \emph{ego}, is at the center of the graph, and the edges connect her to the peers (called \emph{alters}) with which she interacts. Figure~\ref{fig:egonet} illustrates the classical ego network structure. Please note that the second layer includes all the alters in the first layer, the third layer includes all the alters in the first and second layers, and so on. The ego-alter tie strength is typically computed as a function of the frequency of \emph{interactions} between the ego and the alter. The ego network structure is characterised by a striking regularity: the typical sizes of the layers are 1.5, 5, 15, 50, 150, with an approximately constant ratio of 3 between the size of consecutive layers. 
Online relationships have been shown to exhibit similar regularities~\cite{gonccalves2011modeling,Dunbar2015}, thus proving that friendship links (which can be in the thousands for online social networks) not associated with online interactions do not change the innate social structure of humans.

\begin{figure}[t]
\begin{center}
\includegraphics[scale=0.6]{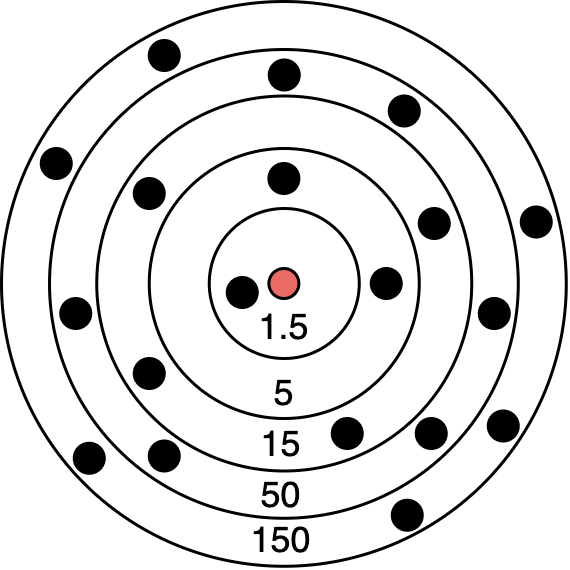}
\caption{Layered structure of the ego network. The red dot at the center corresponds to the ego. The black dots denote the alters, each belonging to a social circle. The typical sizes of social circles (1.5, 5, 15, 50, 150) are also provided.}
\label{fig:egonet}\vspace{-10pt}
\end{center}
\end{figure}



\emph{Motivation.} \ul{Since social links are not all the same, simply considering the set of common neighbors (or analogous topological metrics), disregarding the importance of each social link, might not give a precise estimate of the homophily between two individuals.} In fact, a common best friend is intuitively much more important than a common acquaintance when predicting the formation of new links. Thus, in this paper, we set out to improve feature-extraction link prediction algorithms with information about the different importance of social relationships. \ul{To this end, we exploit well-established models in the anthropology literature that describe the relative importance of our social bonds by grouping them into different social circles.} Please note that the social circles we refer to are very loosely related to the concept of communities in social networks. In fact, social circles are (local) groupings entirely from the perspective of the ego, they include only nodes with which the ego is directly connected, and the links between the alters (which are crucial for regular communities) do not play a role at all here. Note also that our goal is not to propose a new prediction algorithm, but to add social-awareness to popular prediction solutions. 
%
To the best of our knowledge, the only other work in the literature that investigates ego networks in relation to link prediction is~\cite{Aiello2017}. However, in the context of that work, an ego network is a different structure from the one studied here (only the links between the alters are considered), social cognitive models are not leveraged, and the effect of using information on the different social circles for link prediction is not investigated.

While the general problem of link prediction entails recommending any of the useful links that a user may want to create with other users, in certain applications we may want to recommend links between users that share a specific common interest together. Predicting links in such contextualised communities might be much more relevant (and valuable) than predicting links between generic users. This could be the case, for example, of a third party service that wants to connect users interested in a specific topic. In this paper, we showcase one application of this concept by focusing on a set of gamers on Twitter, and we investigate how to make new link recommendations leveraging key properties of the way users allocate their cognitive capacity to social relationships. As discussed in Section~\ref{sec:dataset}, we believe that focusing on such a target set of users, characterised by a specific type of online engagement, is particularly significant for link recommendation in online social networks. Then, we apply the same method to generic Twitter users, in order to assess the efficacy of the proposed solution also beyond the contextualised communities settings.

\emph{Key findings.} The key findings presented in the paper are the following:

\begin{itemize}
    \item In the vast majority of cases, regardless of the prediction approach (unsupervised or supervised), the specific heuristic or learning algorithm, and the metric (precision, AUC, F1 score) considered, leveraging social circles information outperforms the corresponding baseline in which circles are ignored.
    \item The contribution of the different social circles to the achieved precision varies with the recommendation algorithm considered. In unsupervised similarity-based approaches, innermost circles help policies with little or no penalization of high degree nodes, while outermost circles boost the performance of algorithms relying on high-degree nodes penalization. The precision of supervised strategies is generally considerably increased when only the innermost circle, corresponding to the small set (composed of 3-5 people typically) of the most intimate relationships, is considered.
    \item Supervised link prediction that uses circle-aware explicit features is able to beat feature-learning link prediction algorithms like \emph{node2vec} and SEAL, which exploit node embedding and graph neural networks, respectively. This happens when the most intimate circle is leveraged, which comprises only the two or three strongest relationships of the ego. Thus, using knowledge on a few common strong ties is more effective than complex black box approaches.
    %
    %
    \item Leveraging social circles information provides the same performance as using additional classifiers on nodes, which might be impractical or costly to use.
    %

    %
    %
\end{itemize}

The rest of the paper is organised as follows. We review the related work in Section~\ref{sec:relwork}. In Section~\ref{sec:dataset} we present the two datasets that we leverage in our analysis. The ego-aware unsupervised and supervised link prediction approaches based on feature extraction methods that we study, and the feature learning methods that we use as benchmarks are presented Section~\ref{sec:prediction_algorithm}. Their prediction performance is then evaluated in Section~\ref{sec:evaluation}. 
%
Finally, in Section~\ref{sec:conclusion} we conclude the paper. An analysis of the computational complexity of the proposed approach is provided in Appendix~\ref{app:complexity_analysis}\iftoggle{ONEFILE}{}{ of the Supplemental Material (SM)~\cite{toprak2021harnessing-power}}.

\section{Related Work}
\label{sec:relwork}

\subsection{Dunbar's model and ego networks}
\label{sec:relwork_egonetworks}

According to the \emph{social brain hypothesis} from anthropology~\cite{dunbar1998social}, the social life of primates is constrained by the size of their neocortex. Specifically, for humans, the typical group size is estimated around an average of 150 members, a limit that goes under the name of Dunbar's number. This limit is related to the cognitive capacity that humans are able to allocate to nurturing their social relationships. The 150 \emph{friends}, in fact, do not include acquaintances, but only people with which a coherent quality relationship is entertained. In a modern society, this entails, at the very least, exchanging birthday or Christmas cards every year. Unsurprisingly, humans are not fair in how they distribute their cognitive attentions among the 150 important persons they have around them. To the contrary, it is possible to group our social relationships in circles of increasing intimacy, as illustrated in Figure~\ref{fig:egonet}. Typically, each of us has at least four circles of intimacy~\cite{hill2003social,Zhou2005}: the innermost one (\emph{support clique}) includes our close family, the \emph{sympathy group} comprises all those people whose death tomorrow would leave us deeply affected, the \emph{affinity group} includes colleagues, extended family, people you hang out with, and finally the \emph{active network} is made up of all the people you meaningfully interact with at least once a year. The circles are conventionally concentric, with the inner layers contained in the outer ones. Many people also feature an additional inner layer (contained in the support clique) comprising on average $1.5$ alters for which they have a very high emotional investment~\cite{Dunbar2015}. 


The social brain theory was developed for the offline world, in which relationships were nurtured via face-to-face interactions, letters, or landline phone calls. Many had postulated that this theory would not carry over to the online world, where OSN allow us to engage conveniently with a huge number of people scattered across the globe. However, Dunbar's model stood the test of the cyberworld as well: Dunbar's number hold for email communications~\cite{Haerter2012}, mobile phone calls~\cite{Miritello2013}, and Twitter~\cite{gonccalves2011modeling}. More importantly, the same layered structure of social relationships has emerged on both Facebook and Twitter~\cite{Dunbar2015}. 
These findings are extremely important. In fact, since ego network structures are known to impact significantly on the way information spreads in OSN, and on the diversity of information that can be acquired by users~\cite{aral2011diversity}, embedding these models of human cognition into services for OSN may drastically improve the quality of service provided to the users.

\vspace{-10pt}
\subsection{Link prediction}
\label{sec:relwork_linkprediction}

The literature on link prediction algorithms is large. Here we only summarise the main approaches and we refer the interested readers to \cite{Martinez2016,Lu2011,mutlu2020review} for detailed surveys on the topic.
Following the taxonomy proposed by~\cite{mutlu2020review}, we discuss separately feature extraction and feature learning methods


\subsubsection{Feature extraction methods}
\label{sec:relwork_linkprediction_feature_extraction}

Similarity-based methods make up the largest class of link prediction algorithms proposed in the literature. The rationale of this approach is that nodes are more likely to form links with nodes that are similar to them. This idea is grounded in the widespread and well-documented social phenomenon of homophily~\cite{mcpherson2001birds}. Algorithms in this class differ in how they define the similarity between nodes. We can broadly distinguish them based on whether they use local information, global information, or a hybrid combination of the two (this approach is one of the least popular, hence we will not treat it further). For a given node pair, local similarity-based solutions rely on node neighborhood-related structural information, such as the number of common neighbors~\cite{Liben-Nowell2007}, an inverse function of the degree of common neighbors~\cite{Adamic2003,Zhou2009}, or a Preferential Attachment index~\cite{Liben-Nowell2007}. Local similarity-based approaches are very efficient, even on large networks, due to their easy parallelization. Their main theoretical limitation is that they are able to predict only new links between neighbors-of-neighbors. 
Their counterpart, global similarity-based approaches, rely on metrics computed considering the whole network topology. These metrics focus typically on the possible paths inside the network, such as the shortest paths~\cite{liben2005algorithmic}, paths with different lengths by means of the Katz Index~\cite{Liben-Nowell2007}, random walks~\cite{liu2010link}, or community membership~\cite{Soundarajan2012}.


Algorithmic methods map the link prediction problem into well-known algorithmic approaches. Within this category, classifier-based methods treat the link prediction problem as a binary classification problem~\cite{Hasan2006,Zhang2017,Cukierski2011,Zhang2018}. Each node pair can be characterised with a variety of attributes, including the similarity-based heuristics discussed above. This labelled set can be fed to virtually any classifier (such as decision trees, SVM, k-nearest neighbors, random forest, neural networks). The advantage of this approach is that it can be extended and adapted to any new attribute that one wants to test, and that it generally significantly outperforms purely unsupervised similarity-based methods~\cite{Hasan2006,Cukierski2011,Wang2016}.

Preprocessing methods are considered meta-approaches, as they are intended to be used in conjunction with other algorithms, for which they provide preprocessing/pre-filtering intended to remove some noise in the network. As an example, the clustering method discussed in~\cite{liben2005algorithmic} suggests the removal of the weakest links (those between nodes with few or zero common neighbors). The solution that we propose in this paper falls into this category. The meta-approach that we investigate is based on considerations related to how people distribute their social capacity across their relationships, rather than on pure graph-related properties.

Orthogonally to the above classification, we can also distinguish between approaches that consider the weights of the links or not. 
Among the former, we mention~\cite{Zhou2009,Zhu2016}.
As discussed in~\cite{Lu2010}, though, the performance of weighted indices is often worse than their unweighted counterpart. For this reason, in this work we do not consider the weight of links besides what is needed to compute the ego network structure. Besides weights, the signs of links can also be considered for link prediction~\cite{KOU2021106942}.

\subsubsection{Feature learning methods}
\label{sec:relwork_linkprediction_feature_learning}

Feature learning methods map the graph into a low-dimensional feature space. The difference between feature extraction and feature learning methods is that, in the latter, the features are learnt by the system and not hand-engineered. The mapping, or graph representation, can be learned and optimized via both supervised and unsupervised methods~\cite{mutlu2020review}. The random walk methods use graph exploration methods such as breadth-first search (BFS), depth-first search (DFS), and random walks to capture features and node properties such as centrality, being a hub, or community membership. \emph{DeepWalk}~\cite{perozzi2014deepwalk} and \emph{node2vec}~\cite{grover2016node2vec} are the most popular methods of this category. For the link prediction task, the node feature vectors (node embeddings) are transformed into edge feature vectors (edge embeddings) using operators like the Hadamard product or cosine similarity. Extracted edge embeddings are fed to learning algorithms (such as SVM, regression methods, deep neural networks) to train a model that is used to predict future links. These methods are considered semi-supervised, since the information on the existence/non-existence of a link is already present in the studied graph.

Graph neural network (GNN) models leverage neural networks to map the graph structure to a low dimensional vector space. Graph differentiable pooling, graph autoencoders, and graph neural networks have been used for the link prediction algorithms based on neural networks. In this study, we compare against the popular GNN-based link prediction algorithm SEAL~\cite{Zhang2018}, which relies on a CNN architecture (DGCNN, to be specific). Note that network embedding methods and graph neural networks are highly related to each other~\cite{wu2020comprehensive}, as they both map the network to a lower-dimensional representation space. However, GNNs allow for greater flexibility, by leveraging node features in the encoder, by sharing parameters between nodes in the encoder, and by fixing the transductivity problem of shallow embdeddings~\cite{Hamilton2020}. From the link prediction perspective, GNNs typically provide end-to-end frameworks (i.e., the learnt presentation is optimised for the link prediction problem), while graph embedding methods capture topology-level properties that are then used to train a separate classifier for link prediction. 
\section{Data Preprocess}
\label{sec:dataset}

As discussed in Section~\ref{sec:intro}, in this paper we investigate social circle-aware link prediction considering both a community of users with shared interests (gaming, in our case) and generic users. Before presenting the algorithms for predictions in Section~\ref{sec:prediction_algorithm}, here we first introduce the datasets that we use as case study. Without loss of generality, some definitions required in the prediction algorithms will be discussed with reference to these datasets. However, please note that each such definition can be used as-is for any other dataset describing different communities, and the algorithm works unchanged in all such cases.

In the \emph{gaming-related dataset}, we consider as shared topic of interest ``indie games''. An independent video game, or indie game, is a video game that is often created without the financial support of a publisher, although some games funded by a publisher are still considered ``indie''. Indie games often focus on innovation and rely on digital distribution. Growing a community of interested users is thus a critical success factor for them. The goal of our link prediction task is to recommend users interested in indie games to other users interested in indie games. Our gaming-related dataset consists of 8,932 users (labelled as \emph{gamers}) engaging in game-related conversations. We have collected (June 2018) their timeline (most recent 3200 tweets) using the Twitter Search API. Details on the collection process are provided in Appendix~\ref{app:datasets}\iftoggle{ONEFILE}{}{ of the SM~\cite{toprak2021harnessing-power}}.
%
%
In the \emph{generic user dataset}, we focus on a group of users whose interests are not associated with a specific domain. To this aim, we use a dataset presented in~\cite{Arnaboldi2013Dataset}, containing the timelines (most recent 3200 tweets) of 1,930,802 Twitter users obtained by means of a snowball sampling starting from US President Barack Obama in November 2012. Please refer to~\cite{Arnaboldi2013Dataset} for further details on the collection process.

Note that we use two Twitter datasets because Twitter is especially amenable to the computation of ego networks: social interactions are typically public and can be easily downloaded. For this reason, the vast majority of the literature on online ego networks (Section~\ref{sec:relwork_egonetworks}) leverages Twitter data. Facebook interactions would be similarly suitable but they are not public, hence could not be accessed. The social network graphs available to the research community do not typically include information on the frequency of interactions, thus ego networks cannot be extracted.


\subsection{Extracting the ego networks}
\label{sec:extraction_egonets}

We filter the 8,932 gamers and 1,930,802 generic users and extract their ego network (including their social circles) using the same methodology as in~\cite{boldrini2018twitter}, which we briefly recollect here. We simply need a weighted social graph, whose edge weights correspond to the contact frequency between the corresponding nodes. In order to extract reliable information, we filter users according to the same policy used in~\cite{boldrini2018twitter}. We only consider users whose Twitter activity is regular (i.e., they post, on average, at least one tweet every 3 days for at least 50\% of the total number of months of their activity) and stationary (i.e., they are not in their initial stage of engagement with the Twitter platform, which typically features a transient spike of activity).
The strength of ego-alter relationships is inferred from the frequency of direct tweets (mentions, reply, or retweets) between the ego and the alters. The optimal number of social circles per users is obtained using the Mean Shift clustering algorithms to group such frequencies. The code for computing the ego networks can be found at~\url{https://egonetworks.readthedocs.io/en/latest/}.

Approximately one third of the gamers passed all the filters described above, 
and this left us with 3,061 gamers with reliable ego network information. On the other hand, approximately 7\% of the generic users passed all the filters, amounting to 148,105 generic users with reliable ego network information. We further filter them as suggested in~\cite{Liben-Nowell2007} by considering only nodes that are connected to the giant component. Also, in order to have a generic user dataset with a similar size as that of the gaming-related dataset and to reduce its computational intensity, we sampled 3,000 egos from the giant component of the generic users network by snowball sampling~\cite{Leo1961}.

We denote the set of nodes for which we have ego networks as $V_{e}$. In the gaming-related dataset all egos are gamers. Vice versa, in the generic users dataset, all egos are generic users. 
An exploratory analysis of the ego networks in our datasets, which is substantially in agreement with the results from the related literature~\cite{boldrini2018twitter,Dunbar2015}, can be found in Appendix~\ref{sec:results_egonets}\iftoggle{ONEFILE}{}{ of the SM~\cite{toprak2021harnessing-power}}.
%
%
Note that, when we focus on an ego~$i$, its alters can now be associated with the social circle of~$i$ to which they belong. Note also that we do not typically have the ego network of the alters. In the gaming-related dataset, the alters of ego~$i$ can be generic users, other gamers, or even games. This implies that some nodes are \emph{domain-specific} (in our case, gaming-related), while others are generic. Domain-specific nodes are homogeneous, since they share a common interest. Summarising, we have three classes of nodes in the gaming-related dataset: gamer users with ego network ($V_{e}$), regular (unlabelled) users without ego network ($V_n$), and domain-specific users that may or may not have an ego network ($V_{d}$, including games and gamers for which ego network information is not available or not reliable because they didn't satisfy the filters discussed previously).
On the other hand, in the generic user dataset, all alters of ego~$i$ are also generic users.  Generic nodes make prediction more challenging with respect to the case of homogeneous nodes (as, e.g., in~\cite{Liben-Nowell2007}), because they tend to differ more from each other.
In the generic user dataset, we can identify only two classes of nodes: generic users with ego network ($V_{e}$) and generic users without ego networks ($V_{n}$) that are basically alters.

\subsection{Dataset summary}

The final datasets are summarised in Table~\ref{tab:dataset_summary_inactive_gamers} and Table~\ref{tab:dataset_summary_inactive_generic_users}. The imbalance ratio, i.e., the ratio between the negatives (potential links that do not exist in practice) and the positives (links actually existing) in the class of both gamers and generic users is around $10^4$:1, which is quite high. This is a well-known problem in link recommendation for online social networks (remember from Section~\ref{sec:intro} that links are sparse). This implies that special care should be taken in the evaluation, as explained in detail in Section~\ref{sec:metrics}. 


\begin{table}[t]
\caption{Gaming-related dataset }
\label{tab:dataset_summary_inactive_gamers}
\centering
\begin{tabular}{l||r}
\# of gamer nodes $|V_{e}|$ & 2,995 \\
\# of domain-specific nodes $|V_{d}|$ & 70,859 \\
\# of all nodes $|V| = |V_{e} \cup V_{d} \cup V_{n}|$ & 470,485 \\
\# of gamer edges $|E_{e}|$ & 2,614 \\
\# of domain-specific edges $|E_{d}|$ & 154,581 \\
\# of edges $|E|$ & 1,004,011 \\
\end{tabular}
\end{table}

\begin{table}[t]
\centering
\caption{Generic users dataset }
\label{tab:dataset_summary_inactive_generic_users}
\begin{tabular}{l||r}
\# of ego nodes $|V_{e}|$ & 3000 \\
\# of all nodes $|V| = |V_{e} \cup V_{n}|$ & 278,510 \\
\# of ego edges $|E_{e}|$ & 7,158\\
\# of edges $|E|$ & 567,739 \\
\end{tabular}
\end{table}



\section{Link prediction based on social circles}
\label{sec:prediction_algorithm}


We denote with $\mathcal{G} = (V, E)$ the graph modelling the relationships in our dataset. Since our goal is to evaluate the contribution of social circles to link prediction, we will focus on the subset of nodes in $V$ for which social circles can be computed. Thus, predictions will be made for the set of possible edges between nodes in~$V_{e}$.

\subsection{Similarity-based unsupervised learning with circle-awareness}
\label{sec:unsupervised}

In order to showcase the effect of social circles for unsupervised link prediction, we focus on the simple, yet effective class of similarity-based approaches.
Like for all similarity-based link prediction algorithms, the goal of the proposed algorithm is to associate each non-existing link between two users $i,j$ in $V_{e}$ with a score proportional to the likelihood that the link will actually be formed in the future. 
Hence, the goal of the link predictor is to associate with each edge in $V_{e} \times V_{e} - E_{e}$ a probability/confidence that the link will indeed appear in the future.
Qualitatively, this is equivalent to \emph{suggesting users to other users} for possible interactions, friendships, etc. 



In order to study the effect of the social circles on link prediction, we propose to \emph{slice} the social graph based on the membership to a specific circle (recall that our concept of social circles is different from that of communities). Let us focus on a user $i \in V_{e}$. We denote with $C_1^{(i)}, C_2^{(i)}, \ldots$, the sets comprising the neighbors of $i$ that belong to $i$'s first, second,~$\ldots$, social circle, respectively. We have discussed how to obtain these circles in Section~\ref{sec:extraction_egonets}.
Thus, the graph $\mathcal{G}$ sliced according, e.g., to circle $C_3$, is the graph including only links between egos and their alters up to layer 3 (remember that $C_3$ includes $C_1$ and $C_2$). More formally, it can be defined as $\mathcal{G}_{C_3} = (V, E_{C_3})$, where $ E_{C_3} = \{ e_{ij} \in E:  i \in V_e,  j \in C_3^{(i)} \}$. Note that the slicing introduces asymmetry in the graph (e.g., a relationship $i,j$ can be in C4 for ego $i$ and in C3 for ego $j$). However, it is exactly this edge filtering that helps link prediction, as we will show in Section~\ref{sec:evaluation}.

Social-circles slicing can be extended to include also domain-based slicing for domain-specific datasets. For example, from $\mathcal{G}_{C_3}$ we can retain only domain-specific nodes and edges (i.e., nodes in $V_e \cup V_d$ and edges in $E_{e} \cup E_{d}$). We generalise the notation related to the slicing preprocessing by denoting with $\omega$ the specific slicing considered, and with $\mathcal{G}_{\omega}$ the resulting graph. In Section~\ref{sec:evaluation}, we will discuss the performance of link prediction both with social-based slicing and with social-based plus domain-based slicing. We anticipate here that when the social circles information is not used for the prediction, leveraging the category information provides a significant advantage. Vice versa, social-circles awareness makes information on the category less relevant (and this is extremely important, since the extraction of categories typically requires an additional, domain-specific, classifier or manual labelling).


In order to perform unsupervised link predictions, for each user-user pair $i,j$ for which a link does not exist in $E_{e}$, we follow Definition~\ref{def:unsupervised_circle_aware} below. Note that each slicing will yield different predicted links. Hence, we will treat each slicing choice as a separate link prediction approach.
%
%
\begin{definition}[\textsc{Unsupervised Circles-aware Link Prediction}]\label{def:unsupervised_circle_aware}
For a fixed K, the social circles-aware link prediction algorithm suggests, to a user~$i$, K users $j$ (with $j \in V_{e} \wedge (i,j) \not\in E_{e}$) associated with the top-K $sim_{\omega}(i,j)$ values,
where $\omega$ denotes the selected slicing and $sim_{\omega}(i,j)$ is the similarity computed on $\mathcal{G}_{\omega}$.
\end{definition}

The definition of a new similarity function is out of the scope of the paper. What we want to do here is to evaluate popular similarity functions available in the literature when they are enriched with knowledge about the ego network circles. In order to isolate the effect of social-circles awareness, we need approaches that are simple (so it is easy to gauge the role of social circles) yet effective. For these reasons, our choice fell on the strategies described below, often used in the related literature. In summarising them, we denote with $\Gamma_{\omega}(i)$ the neighborhood of node $i$ in $\mathcal{G}_{\omega}$.




\begin{itemize}
\item \emph{Common neighbors (CN)} \cite{Liben-Nowell2007}: the similarity is given by the number of common neighbors between users $i$ and $j$ in $\mathcal{G}_{\omega}$:
\begin{equation}
sim_{\omega}(i,j) = | \Gamma_{\omega}(i) \cap \Gamma_{\omega}(j) |.
\end{equation}

\item \emph{Jaccard's Coefficient (JC)} \cite{Liben-Nowell2007}: the similarity is computed as the Jaccard similarity of the set of common neighbors of users $i$ and $j$ in $\mathcal{G}_{\omega}$:
\begin{equation}
sim_{\omega}(i,j) = \frac{| \Gamma_{\omega}(i) \cap  \Gamma_{\omega}(j)|}{| \Gamma_{\omega}(i) \cup  \Gamma_{\omega}(j) |}.
\end{equation}

\item \emph{Adamic-Adar (AA)} \cite{Adamic2003}: the Adamic-Adar similarity reduces the importance of common neighbours having high degree:
\begin{equation}
sim_{\omega}(i,j) = \sum_{z \in \Gamma_{\omega}(i) \cap \Gamma_{\omega}(j)} \frac{1}{\log(| \Gamma(z) | )}.
\end{equation}

\item \emph{Resource Allocation (RA)} \cite{Zhou2009}: the RA score is similar to the Adamic-Adar one, but it penalises even more the common neighbors with high degree:
\begin{equation}
sim_{\omega}(i,j) = \sum_{z \in \Gamma_{\omega}(i) \cap \Gamma_{\omega}(j)} \frac{1}{| \Gamma(z) |}.
\end{equation}
\end{itemize}

\noindent
Among these four link prediction algorithms, RA is the one that consistently performs better in the related literature~\cite{Lu2010,Zhang2015,Zhu2016}. In Sections~\ref{sec:results_unsupervised} and~\ref{sec:results_supervised} we will investigate if this is still the case when the policies can leverage the knowledge of the social circles.

\subsection{Supervised learning with social-circle awareness}
\label{sec:supervised}

Each topological metric discussed in the previous section captures a single possible mechanism yielding to the formation of new links in the network. Ref.~\cite{Lichtenwalter2010} discusses how this unsupervised approach can take advantage of adding supervised learning on top of it. They show that simple ranking of heuristics is outperformed by supervised classifiers, because the latter have the capability of identifying multiple differentiating boundaries in the similarity score domain, even when just a single heuristic is used as feature. 

Thus, in this section, we cast our social-aware link prediction problem into a supervised learning problem. This entails computing a vector of features for each user pair~$i,j$. Each of the metrics in Section~\ref{sec:unsupervised} becomes a feature that describes the user pair. Each pair is also labelled to mark whether the link exists or not. In order to test the different performance of different supervised approaches, in this work we consider the following learning algorithms (in parentheses we report the link prediction papers in which they have been previously used): logistic regression~\cite{OMadadhain2005}, Random Forest~\cite{Lichtenwalter2010}, decision trees, na\"ive Bayes, and SVM~\cite{Hasan2006}, the latter with both linear and polynomial kernels. We use their R implementations (\texttt{glm}, \texttt{randomForest}, \texttt{rpart}, \texttt{klaR}, \texttt{kernlab}, respectively) together with the \texttt{caret} package for training and test. Parameter optimization is applied on the training set with 10-fold cross validation.

\subsection{Prediction based on feature learning methods}
\label{sec:latent_features_based}

In the literature, there are newly suggested link prediction methods based on learning \emph{latent features} of the graphs, where these feature vectors are low-dimensional vector representations produced by approaches such as graph representation learning~\cite{cai2018comprehensive} and graph neural networks (GNN) \cite{wu2020comprehensive}.
For a detailed description and classification of the methods, the reader may refer to \cite{wu2020comprehensive, mutlu2020review, kumar2020link, cai2018comprehensive} and Section~\ref{sec:relwork_linkprediction_feature_extraction}. In this work, we select two feature learning algorithms to be used as benchmarks for comparison with our social-aware approach. 
Specifically, we have selected one of the most popular graph embedding methods, \emph{node2vec}~\cite{grover2016node2vec}, and one most the most popular GNN-based link prediction methods, \emph{SEAL}~\cite{Zhang2018}. For details on the two algorithms and their settings, please refer to Appendix~\ref{app:feature_learning}\iftoggle{ONEFILE}{}{ of the SM~\cite{toprak2021harnessing-power}}. Note that, differently from Sections~\ref{sec:unsupervised}-\ref{sec:supervised}, the selected feature learning approaches are not modified to include social-awareness. In fact, the key idea of feature learning is to autonomously learn the important graph features. Thus, our objective, in this case, is to assess whether the automatically learnt features are better than the social-aware hand-engineered ones in predicting new links.

\section{Evaluation}
\label{sec:evaluation}

In this section we carry out the performance evaluation of the link prediction approach proposed in Section~\ref{sec:prediction_algorithm} using the Twitter datasets described in Section~\ref{sec:dataset}.
\subsection{Training and test data}
\label{sec:train_test}


For obtaining longitudinal data, we downloaded (December 2019) the Twitter timeline of the gamers in $V_e$ one year and a half after the initial download, and timeline of the generic users in $V_e$ eight years after the initial download. We then identified the links between gamers in the gaming-related dataset and between generic users in the generic users dataset that have appeared in the meanwhile. These new 843 links between gamers and new 1216 links between generic users constitute the set $E_{new}$ of links to be predicted (i.e., the test set), while $E_{old}$ contains the links existing in the first temporal snapshot (corresponding to the $E_e$ sets in Tables~\ref{tab:dataset_summary_inactive_gamers} and~\ref{tab:dataset_summary_inactive_generic_users}, so $E_{old}=E_e$ for each dataset). The edges in $E_{old}$ are used as positives for computing the heuristics in unsupervised link prediction and as positive training set for supervised link prediction.

While the above considerations are sufficient for unsupervised link prediction, with supervised learning we also need to assign the negatives (i.e., missing links) to the train and test sets.
To this aim, we split them 90\%-10\% between train and test.
Clearly, the set of edges selected for test set may offer only a partial view of the performance (i.e, they may be easier or harder than average to predict).
For this reason, we perform a $k$-fold cross-validation~\cite{Lu2011}, with $k=10$, each time selecting a new 10\% of the negative links for the test set.
The performance metrics are aggregated via microaveraging~\cite{Forman2010}.

Using the 90\% of negatives for training entails training supervised schemes on millions of negatives and only thousands of positives (Table~\ref{tab:dataset_summary_inactive_gamers},  Table~\ref{tab:dataset_summary_inactive_generic_users}).
While the impact of class imbalance is typically limited for similarity-based unsupervised approaches like the ones discussed in Section~\ref{sec:unsupervised}, it can create serious problems for the scalability and reliability of supervised learning.
In order to mitigate the problem, the common approach in the related literature~\cite{Hasan2006,Wang2011,Zhang2018} is to undersample the negative class.
Despite its widespread use, this technique is not without drawbacks~\cite{Yang2015}.
In order to provide reliable measurements, we will show both the results obtained with undersampling (Section~\ref{sec:results_supervised}) and the results obtained with the complete negative class (Appendix~\ref{app:supervised_additional}\iftoggle{ONEFILE}{}{ of the SM~\cite{toprak2021harnessing-power}}).
\subsection{Evaluation metrics}
\label{sec:metrics}

Similarly to the related literature~\cite{Liben-Nowell2007,Zhu2016,Lu2011}, unsupervised link prediction algorithms are evaluated using a top-K analysis, i.e., we compare their performance in predicting $K$ new links. This allows for a fair comparison among the different approaches, as it avoids fixing a  similarity threshold for approaches in which the same threshold may have a different meaning. In order to span a reasonable $K$ range around the number of positives in the test set (which, as we discussed in Section~\ref{sec:train_test}, are 843 and 1216,  for the gaming and generic users datasets respectively), we consider $K \in \{100, 843, 1000\}$ for the gaming-related dataset and $K \in \{100, 1216, 1500\}$ for the generic users dataset.
For a given slice $\omega$ and for a fixed~$K$, each unsupervised link predictor described in Section~\ref{sec:unsupervised} outputs a  list $L_{\omega}^{K}$ of $K$ pairs in $V_{e} \times V_{e} - E_{old}$, which are the newly predicted links (each with its associated confidence in prediction, which, in our case, is the similarity value).
Then we can compute classic metrics such as:

\begin{itemize}
\item $TP = | L_{\omega}^{K} \cap E_{new} |$
\item $FP = | L_{\omega}^{K} - E_{new} |$
\item $FN = | E_{new} - L_{\omega}^{K}  |$
\item $TN = | V_{e} \times V_{e} - E_{old}| - | L_{\omega}^{K} \cup E_{new}|$.
\end{itemize}

Supervised approaches automatically label all the elements in the test set (without the need to set the $K$ or a similarity threshold), hence the number of positives predicted is part of the workings of the supervised approach.
In this case, thus, we do not pick the top $K$ potential edges, but we directly rely on the labelling of the trained predictor for the test edges.
It is then straightforward to compute the above metrics (for example, $TP$s are those edges in $E_{new}$ that are marked as \emph{new} by the supervised prediction algorithm).


The related literature on evaluating systems with class imbalance suggests using metrics like precision ($=\frac{TP}{TP + FP}$), recall (aka TPR) and $F_{1}$ score. As already argued by~\cite{Wang2011}, in the context of link prediction, precision is more important than the other metrics (such as recall), because if the precision is high, one can live with some false negatives. Therefore, the focus of this evaluation will be mostly on precision. For the completeness of results, we also provide, when relevant, the AUC of the precision-recall curve. The AUC provides an aggregate measure of performance across all possible $K$ classification threshold, hence it complements the results obtained for fixed $K$s. The results for the F1 score can be found in Appendix~\ref{app:supervised_additional}\iftoggle{ONEFILE}{}{ of the SM~\cite{toprak2021harnessing-power}}. 

Precision and recall are evaluated based on metrics (TP, FP, TN, FN) which corresponds to specific realizations of a target phenomenon.
Thus, the question arises of how confident we can be about the results obtained on a particular collection that is the result of random sampling (as in the case of cross-validation or subsampling of negatives for supervised prediction). To quantify this confidence, we use credible intervals~\cite{Goutte2005}. For further details, please refer to Appendix~\ref{app:evaluation_setup}\iftoggle{ONEFILE}{}{ of the SM~\cite{toprak2021harnessing-power}}.

\subsection{Experimental setup}
\label{sec:evaluation_setup}
We are interested in comparing the performance obtained using the baseline approach (no social circles information, all edges are considered) and the other strategies in which only edges belonging to the ego network circles are considered.
The majority of egos in our datasets feature approximately five circles in their ego networks (Appendix~\ref{sec:results_egonets}\iftoggle{ONEFILE}{}{ of the SM~\cite{toprak2021harnessing-power}}, Figure~\ref{fig:optimal_circle_distr}).
For this reason, in our circle-based slicing we will consider circles from C1 to C5, then we group together into the \textsc{Active} circle all the circles beyond C5 (for those nodes that have more than five circles).
Note that, since social layers are concentric and the $i$-th also include the $(i-1)$-th up to the first one, for egos with less than five circles the predictions in, e.g. C5, are simply based on the last non-empty social layer (which includes all the innnermost ones).
Recall that the alters in the \textsc{Active} circle are all those with a \emph{significant} relationship with the ego, i.e, they interact with the ego at a frequency of at least one contact per year~\cite{Zhou2005}.
We denote with $\textsc{All}$ the situation in which all relationships are considered, which implies considering both the nodes in the \textsc{Active} circle (significant bonds) and those outside (acquaintances). \emph{Note that, for each strategy, the $\textsc{All}$ case corresponds to the baseline from the related literature}: all relationships are treated as equal, without factoring in the role of social circles.
As anticipated in Section~\ref{sec:prediction_algorithm}, we will consider two scenarios. In the first one, denoted as \textsc{AllEdges}, we consider the full graph $\mathcal{G}$, slicing it based on the social circles both in the gaming-related dataset and the generic dataset. In the second one (\textsc{DomainSpecificEdges}), we perform both social-based and category-based slicing, by only retaining edges between nodes that are domain-specific (i.e., gamers and games, in our case) with the gaming-related dataset.

\subsection{Evaluation of unsupervised link prediction}
\label{sec:results_unsupervised}

We start the evaluation with the unsupervised case. Note that credible intervals are not provided here, because in the unsupervised case there is no sampling of negatives and the links to be predicted are, deterministically, those that have actually appeared between the first and second download.

\begin{figure}
\centering
\subfloat[Gaming-related dataset
\label{fig:allk_precision_alledges_a}]
{\includegraphics[width=0.35\textwidth]
{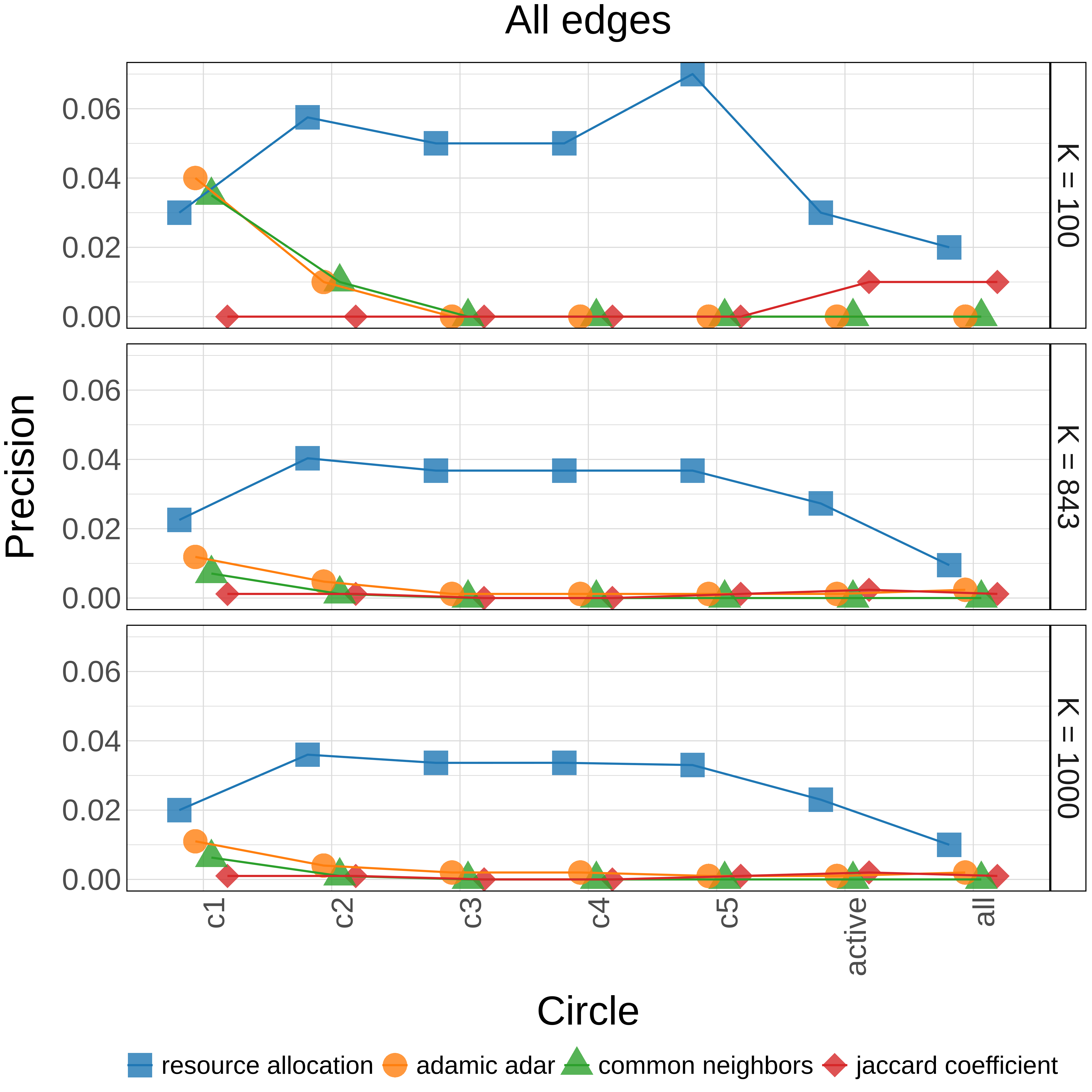}}
\hfill
\subfloat[Generic users dataset.
\label{fig:allk_precision_alledges_b}]
{\includegraphics[width=0.35\textwidth]
{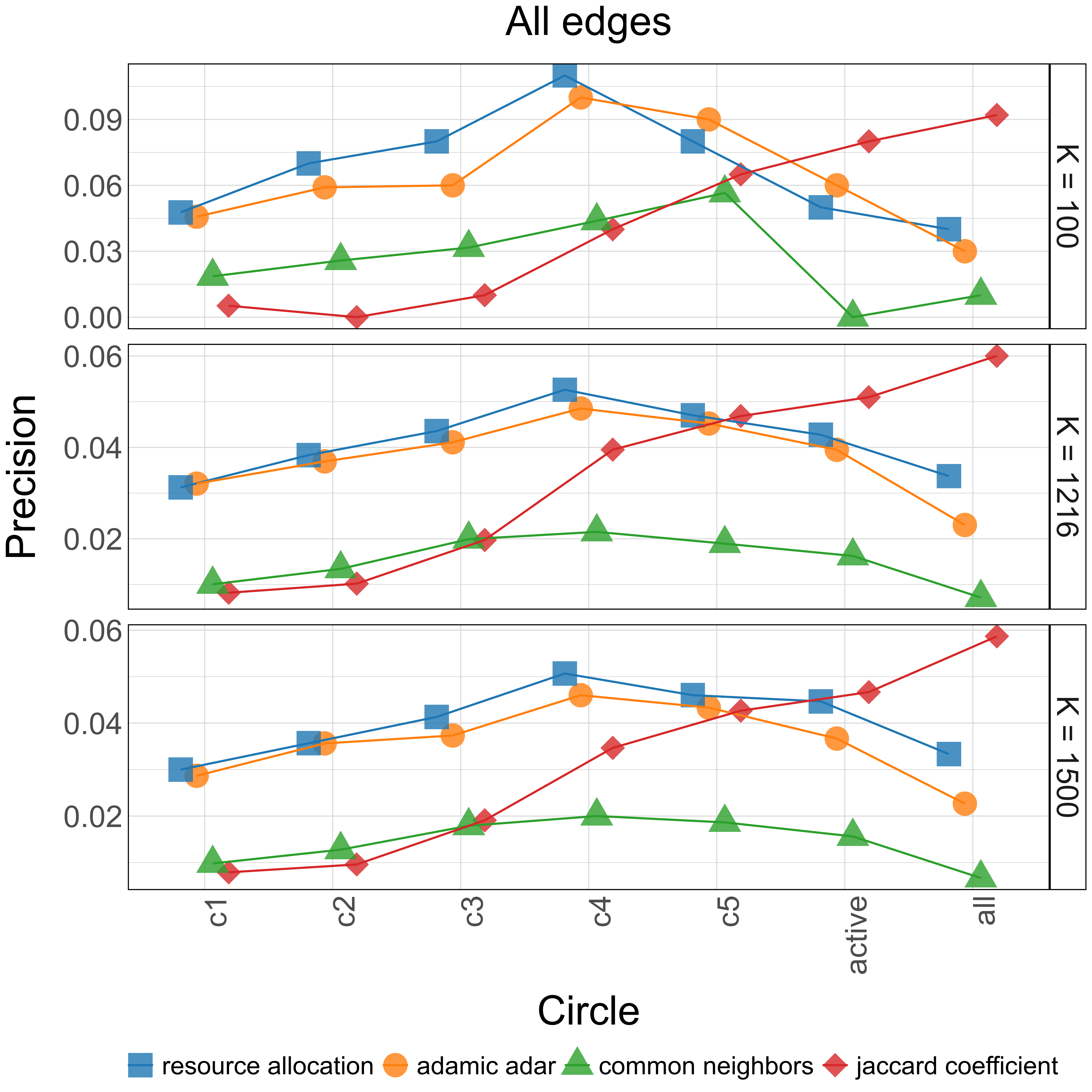}}
\hfill
\caption{Precision in unsupervised settings - \textsc{AllEdges}.} \vspace{-10pt}
\label{fig:allk_precision_alledges}
\end{figure}

%

\subsubsection{The \textsc{AllEdges} case with the gaming-related dataset}
The precision\footnote{Note that the generally low precision values are due to the high class imbalance in the dataset. The benchmark random predictor (the always positive one), in the same settings, would achieve an even worse performance, approximately equal to $2*10^{-4}$.} for varying $K$ (number of new links recommended) is shown in Figure~\ref{fig:allk_precision_alledges_a}.
The maximum precision across all policies decreases when we increase $K$: the more the links recommended, the more the mistakes (as it is generally the case in the related literature).
%
RA clearly outperforms the others in the gaming-related dataset, and this is consistent with previous findings in the related literature~\cite{Lu2010}.
%
%
RA seems to benefit significantly from circle-awareness.
Specifically, ignoring the outer circle gives a clear advantage to RA, which at least doubles its precision for all $K$ values. Thus, the mechanism whereby high-degree penalization and social-awareness work hand in hand is very effective from a link prediction standpoint. It is also interesting to note that this advantage seems lost when only the innermost circle is considered. However, leveraging only the innermost circle C1 means using very little information (only a few links per ego, as can be seen in Figure~\ref{fig:c5_alter_distr}, Appendix~\ref{sec:results_egonets}\iftoggle{ONEFILE}{}{ of the SM~\cite{toprak2021harnessing-power}}). This can be an advantage in case of resource limitations (in terms of computational time, computing the similarity is faster on smaller neighborhoods, as we discuss in Appendix~\ref{app:complexity_analysis}\iftoggle{ONEFILE}{}{ of the SM~\cite{toprak2021harnessing-power}}). In such cases, we can have a good prediction (still better than that at $\textsc{All}$) by using only C1. Thus, the predictive power of using only the most intimate relationships unexpectedly outperforms baselines relying on \emph{all} relationships.
%
The drastically different performance of AA with respect to RA in Figure~\ref{fig:allk_precision_alledges_a} tells us that the way we penalise high-degree nodes impacts on the effect that social circle information has on link prediction effectiveness.
However, it is interesting to note that while in general RA is better than AA, this is not the case when resource constraints are considered.
When only the most intimate links are kept in the ego networks, AA performs at its best and its precision is not only equivalent to that of RA, but also comparable to the precision of RA in the baseline (\textsc{All}). In this case, then, strong degree penalization seems to be less important than strong intimacy.
%
%
The Common Neighbors policy performs quite similarly to AA, even if, in CN, all neighbors contribute the same to the similarity score, regardless of their degree. This is evidence that here AA is working in a regime in which all nodes weight approximately the same (i.e., its logarithmic degree penalization is not enough). Given the similarity with AA, the same considerations we made for AA hold for CN.
%
Finally, we highlight the poor precision of the Jaccard-based strategy, for all $K$s and for all circles in the gaming-related dataset. Jaccard similarity gives more weight to neighborhoods that are very similar to each other. This means that not only the overlapping part is considered, but also the number of nodes that are not in common (union of neighbors). As testified by the plots, this restriction does not give an advantage for link prediction in the gaming-related dataset. With respect to social-circles awareness, Jaccard seems to suffer from using only the intimate relationships, while using the \textsc{Active} layer yields the same precision as the the baseline \textsc{All} while saving computational resources.

We report in Table~\ref{tab:auc_alledges_liga} the Area Under the prediction-recall Curve.
Recall that the AUC is not dependent on $K$, since it is obtained by exploring the whole range of similarities for making recommendations.
When considering the relation between precision and recall through the AUC, we observe that the best overall results (values underlined in the table) are achieved by RA that leverages C4 and C5.
Looking at the layers that, for each policy, provide the best AUC, we note that the \textsc{All} case (which is the baseline for similarity-based policies from the related literature) is never the best performing. From the circle-awareness standpoint, there seem to exist two classes of policies, those that benefit most from using the outermost circles and those that benefit from the innermost ones.

\begin{table}[t]
\centering
\caption{\textsc{AllEdges} scenario with gaming-related dataset: AUC ($\times 10$)  for the precision-recall curve of unsupervised link prediction. \textmd{We highlight in \textbf{bold} the social layer in which each policy performs at its best. We \underline{underline} the best two AUC overall.}}
\footnotesize
\resizebox{0.48\textwidth}{!}{\begin{tabular}{@{}llllllll@{}} \toprule
 & \multicolumn{7}{c}{\textbf{AUC} ($\times 10$)} \\  \cmidrule{2-8}
 & C1 & C2 & C3 & C4 & C5 & \textsc{active} & \textsc{all} \\ \midrule
RA & 0.01179 & 0.03162 & 0.04699 & \underline{\textbf{0.05364}} & \underline{0.05162} & 0.03641 & 0.01916 \\
AA & \textbf{0.02148} & 0.00749 & 0.01023 & 0.01169 & 0.01270 & 0.01372 & 0.01122 \\
CN & \textbf{0.01914} & 0.00540 & 0.00764 & 0.00902 & 0.01000 & 0.01144 & 0.00924 \\
JC & 0.00332 & 0.00448 & 0.00648 & 0.00813 & 0.00927 & \textbf{0.01007} & 0.00969 \\
\bottomrule
\end{tabular}}
\label{tab:auc_alledges_liga}
\vspace{-20pt}
\end{table}

\subsubsection{The \textsc{AllEdges} case with the generic users dataset} Similarly to the gaming-related dataset, we observe that the maximum precision across all policies decreases when we increase $K$ (Figure~\ref{fig:allk_precision_alledges_b}). In this case, though, the baseline results (i.e., those without social-awareness) are different: while, for the gamers network, RA was already the best predictor in the \textsc{All} case, here we observe a generalised advantage of the JC approach for all $K$ values. Given this different starting point, let us study how the prediction policies react to circle-awareness.
RA preserves the characteristics we observed for the gamers, and still very much benefits from social-awareness: all the circle-based scenarios have better or comparable precision than the baseline \textsc{All}. Note that, as mentioned in the gaming-related dataset results, even if the prediction precision is similar -- as it is the case in Figure~\ref{fig:allk_precision_alledges_b} for C1 and \textsc{All} with large $K$ -- using the most intimate links in C1 gives us the advantage of smaller computation time against traditional methods (see Appendix~\ref{app:complexity_analysis}\iftoggle{ONEFILE}{}{ of the SM\cite{toprak2021harnessing-power}} for a thorough discussion).
%
Turning our attention from RA to AA, Figure~\ref{fig:allk_precision_alledges_b} shows that AA performs significantly better compared to what we observed for the gaming-related dataset: here, AA substantially mimics the performance of RA, both in terms of achieved precision and in terms of social circles in which it performs best (specifically, around C4). Interestingly, AA was mimicking instead the performance of CN in the gaming-related dataset, with peak performance in C1 and a generally poor precision both in the other social circles and in the baseline \textsc{All}. The fact that this difference is due to the different structural properties of the generic users graph with respect to the gamers one follows directly from the definitions of RA, AA, and CN in Section~\ref{sec:unsupervised}. Indeed, all the three policies are based on the principle that the more the common neighbors, the better for link prediction. Then, in CN all common neighbors weight the same, while in AA and RA the common neighbors with high degree weight less (much less in RA than in AA, since in RA the penalization is linear while in AA it is logarithmic) than those with small degree. The fact that the performance of AA is approaching that of RA implies that the degree of the common neighbors is not very high (hence the logarithmic penalization of AA and the linear one of RA are in a regime in which they yield similar scores). Taking into account the degree of the common neighbors still provides an advantage, though, as highlighted by the worse performance of CN with respect to AA and RA.
%
As a final remark, note that while JC is performing extremely well with the generic users dataset in the baseline \textsc{All}, its performance rapidly deteriorates as we incorporate circle information. This tendency was already present in the gamers datasets (Figure~\ref{fig:allk_precision_alledges_a}) but was somewhat masked by the generally poor precision achieved by JC regardless of the circle/baseline considered.
It is interesting to speculate on why circle-awareness does not help JC in general. To this aim, let us consider the definition of the Jaccard similarity in Section~\ref{sec:unsupervised}. JC captures the fraction of common neighbors out of all neighbors for a pair of users $i,j$. Thus, its very nature requires ``useless'' neighbors (i.e. those that are not in common) to be present, in order to differentiate between high and low neighborhood overlap. However, the slicing based on social-circles effectively remove ``useless'' nodes, hence impairing the discriminating capabilities of the JC score.

In Table~\ref{tab:auc_alledges_generic_users}, we report the Area Under the prediction-recall Curve for the generic uses dataset. Again, recall that the AUC captures the trade-off between precision and recall and that it is not dependent on the specific $K$ value considered, as it spans the whole $K$ range. The best overall precision-recall performance is achieve by RA when it leverages social information in C5. The second best AUC is provided by JC in the baseline case (no circle-awareness). AA and CN provides the highest AUC when leveraging C5 information.
Overall, from the results in Table~\ref{tab:auc_alledges_generic_users}, we can thus conclude that, even though the advantage of circle awareness is less evident in the generic users dataset, it still outperforms non-circle-related approaches.

\begin{table}[t]
\centering
\caption{\textsc{AllEdges} scenario with generic users dataset: AUC ($\times 10$)  for the precision-recall curve of unsupervised link prediction. \textmd{We highlight in \textbf{bold} the social layer in which each policy performs at its best. We \underline{underline} the best two AUC overall.}}
\footnotesize
\resizebox{0.48\textwidth}{!}{\begin{tabular}{@{}llllllll@{}} \toprule
 & \multicolumn{7}{c}{\textbf{AUC} ($\times 10$)} \\  \cmidrule{2-8}
 & C1 & C2 & C3 & C4 & C5 & \textsc{active} & \textsc{all} \\ \midrule
RA & 0.03028 & 0.05755 & 0.10289 & 0.14233 & \underline{\textbf{0.15142}} & 0.13941 & 0.09313 \\
AA & 0.02730 & 0.04987 & 0.07930 & 0.11554 & \textbf{0.12075} & 0.10057 & 0.06273 \\
CN & 0.01106 & 0.01914 & 0.03093 & 0.04678 & \textbf{0.04714} & 0.04171 & 0.03088 \\
JC & 0.00927 & 0.01737 & 0.03573 & 0.07824 & 0.10735 & 0.12498 & \underline{\textbf{0.14390}} \\
\bottomrule
\end{tabular}}
\label{tab:auc_alledges_generic_users}
\vspace{-15pt}
\end{table}

\subsubsection{The \textsc{DomainSpecificEdges} case}
Recall that, in this case, for any social-based slicing, we also remove all nodes (and associated edges) that are not domain-specific (this is only relevant to the gaming-related dataset). All the nodes left after this additional, category-based slicing are either gamers or games, i.e., they belong to the specific community for which we are making link recommendations. Making a link prediction is now much easier, because the network has been pruned by nodes potentially irrelevant or misleading. Indeed, as shown in Figure~\ref{fig:allk_precision_domainedges_liga}, the precision of AA and CN significantly improves with respect to the \textsc{AllEdges} scenario in Figure~\ref{fig:allk_precision_alledges_a}. This improvement is also present when no social-circle information is used (\textsc{All} circle in Figure~\ref{fig:allk_precision_alledges_a}). This implies that when social-circle awareness is not used for  prediction, leveraging the category information provides a significant advantage. Vice versa, social-circles awareness seem to make information on the category unnecessary. Since the extraction of user categories typically requires a domain-specific classifier or manual labelling, being able to skip this phase without affecting the link prediction quality would be very important.
The precision of the Jaccard-based approach remains very low also when considering domain-specific edges only.
%
We conclude this part on the \textsc{DomainSpecificEdges} scenario by analysing the AUC of the precision-recall curve, reported in Table~\ref{tab:auc_domainedges_liga}. Table~\ref{tab:auc_domainedges_liga} shows that, also in this case, the best performance is achieved by RA in the outermost circles (specifically, \textsc{Active} and C5). We observe again the duality of behaviours with respect to the gain from social-circles awareness: RA and AA achieve their best AUC in \textsc{Active}, CN and JC in C1. The baseline \textsc{All} is never optimal (neither from an absolute standpoint nor for specific policies). This further confirms the advantage of using social circle awareness in general.

\begin{figure}[t]
\centering
  \includegraphics[width=0.35\textwidth]{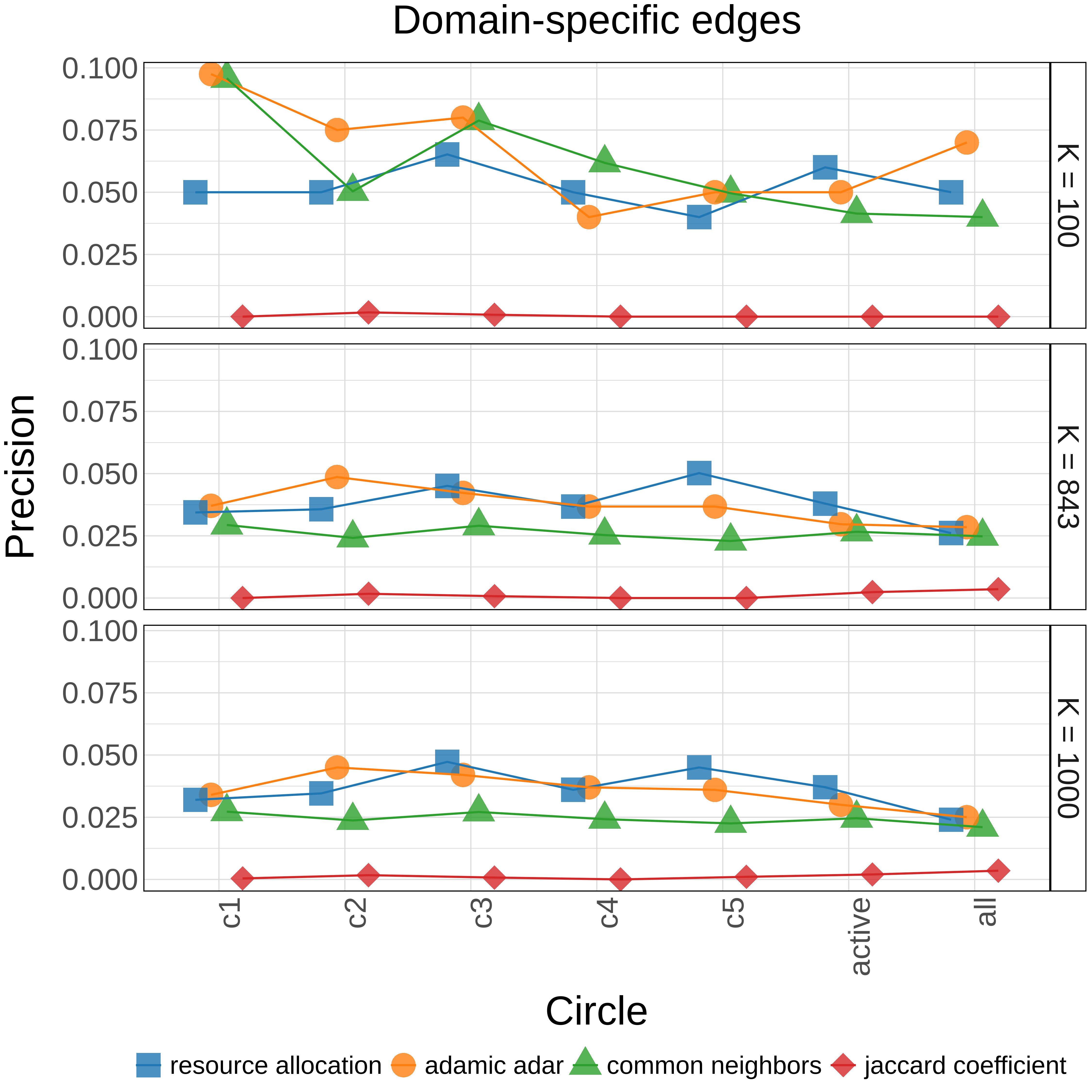}
  \caption{Precision in unsupervised settings - \textsc{DomainEdges} (only relevant to the gaming-related dataset).}
  \label{fig:allk_precision_domainedges_liga}
\end{figure}


\begin{table}[t]
\centering
\caption{\textsc{DomainEdges} scenario with gaming-related dataset: AUC ($\times 10$)  for the precision-recall curve of unsupervised link prediction. \textmd{We highlight in \textbf{bold} the social layer in which each policy performs at its best. We \underline{underline} the best two AUC overall.}}
\footnotesize
\resizebox{0.48\textwidth}{!}{\begin{tabular}{@{}llllllll@{}} \toprule
 & \multicolumn{7}{c}{\textbf{AUC} ($\times 10$)} \\  \cmidrule{2-8}
 & C1 & C2 & C3 & C4 & C5 & \textsc{active} & \textsc{all} \\ \midrule
RA & 0.3331 & 0.5457 & 0.7275 & 0.7361 & \underline{0.8175} & \underline{\textbf{0.8807}} & 0.6605 \\
AA & 0.6564 & 0.5754 & 0.7638 & 0.6762 & 0.6834 & \textbf{0.7870} & 0.5745 \\
CN & \textbf{0.6523} & 0.5643 & 0.5381 & 0.4289 & 0.4131 & 0.5562 & 0.4028 \\
JC & \textbf{0.1296} & 0.1270 & 0.1353 & 0.1205 & 0.1117 & 0.1053 & 0.0950 \\
\bottomrule
\end{tabular}}
\label{tab:auc_domainedges_liga}
\end{table}

\vspace{-10pt}
\subsection{Evaluation of supervised link prediction}
\label{sec:results_supervised}
In this section we assess the advantages brought about by including a supervised classifier in the link prediction approach. As discussed in Section~\ref{sec:supervised}, our classifier uses as features the similarity-based metrics introduced in Section~\ref{sec:unsupervised}. The supervised approaches that we test are logistic regression, decision trees, na\"ive Bayes, Random Forest, SVM. Recall (from Section~\ref{sec:train_test}) that with supervised learning we investigate the performance both on the original dataset and on the undersampled one. Also, with supervised learning we do not carry out a top-K analysis but we allow the algorithms to freely classify all the edges in the test set.

\begin{figure}[t]
\centering
\subfloat[Gaming-related dataset
\label{fig:supervised_precision_undersampled_alledges_a}]
{\includegraphics[width=0.4\textwidth]
{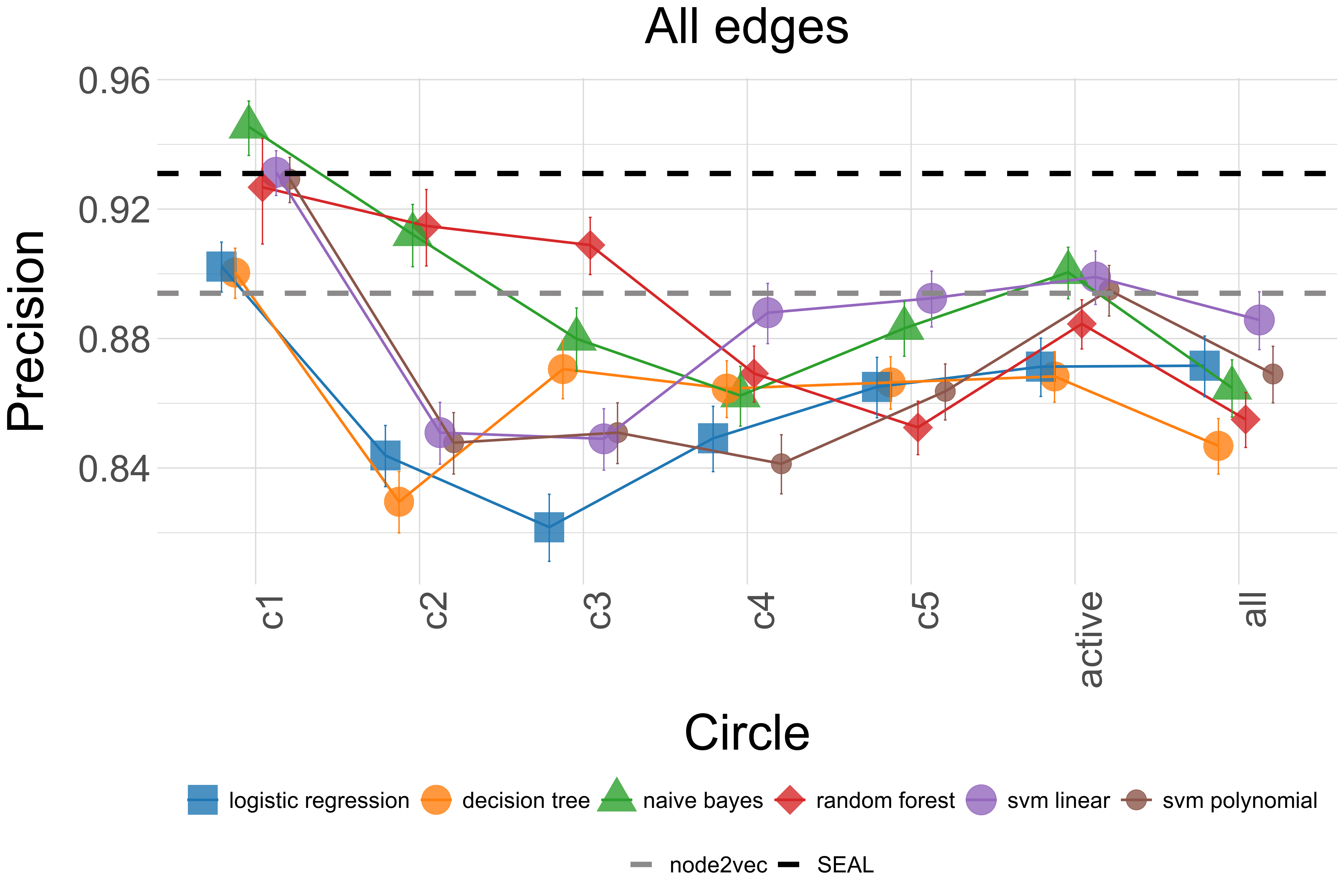}}
\hfill
\subfloat[Generic users dataset.
\label{fig:supervised_precision_undersampled_alledges_b}]
{\includegraphics[width=0.4\textwidth]
{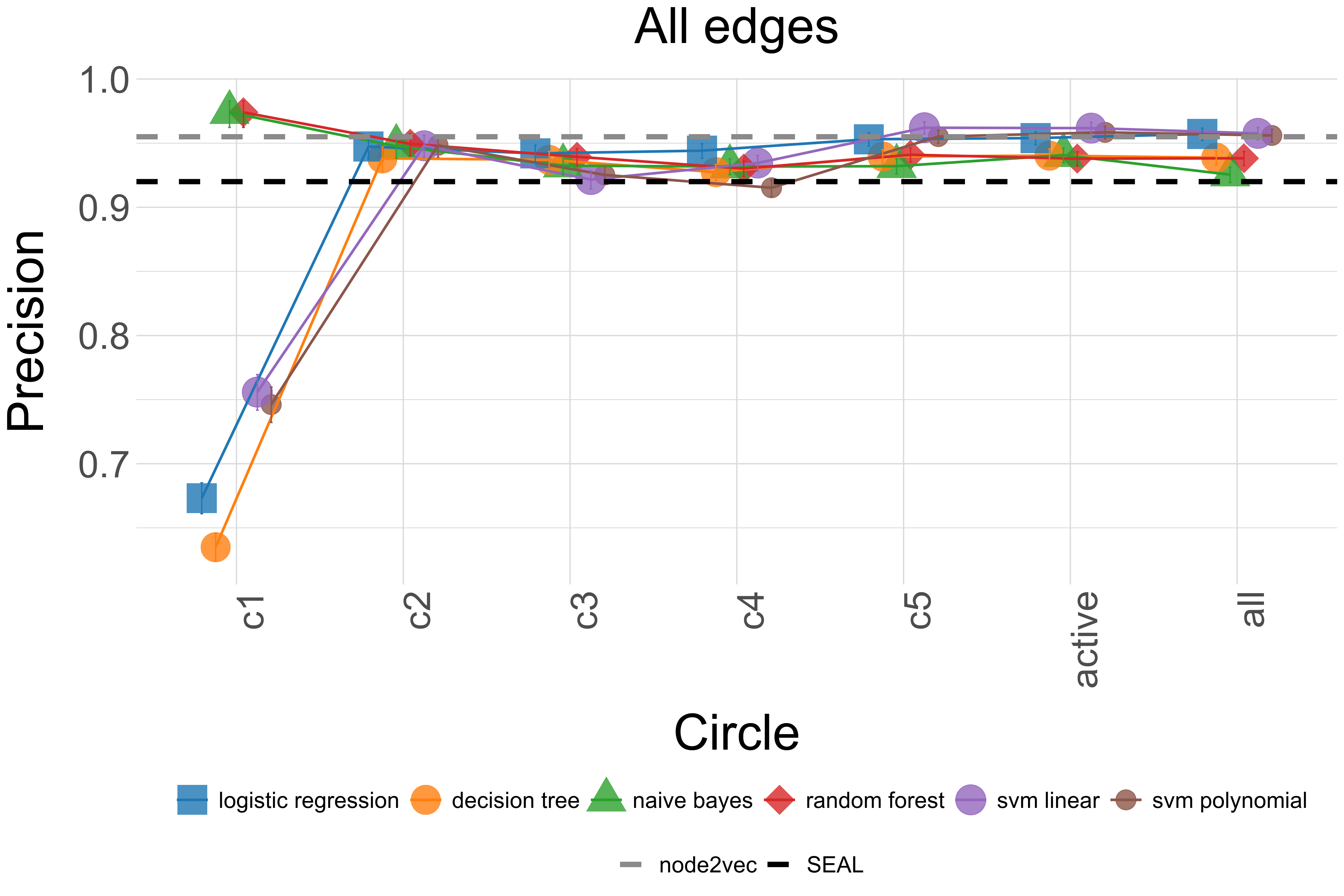}}
\hfill
\caption{Precision (with credible intervals) - \textsc{Supervised AllEdges} on the undersampled datasets. Note the different y-axis range between the two panels (a) and (b). The $\textsc{All}$ scenario of feature learning algorithms is also included (it will be discussed in Section~\ref{sec:networkembeddings}).}
\label{fig:supervised_precision_undersampled_alledges}
\end{figure}

\begin{figure}[t]
\centering
\includegraphics[width=0.35\textwidth]{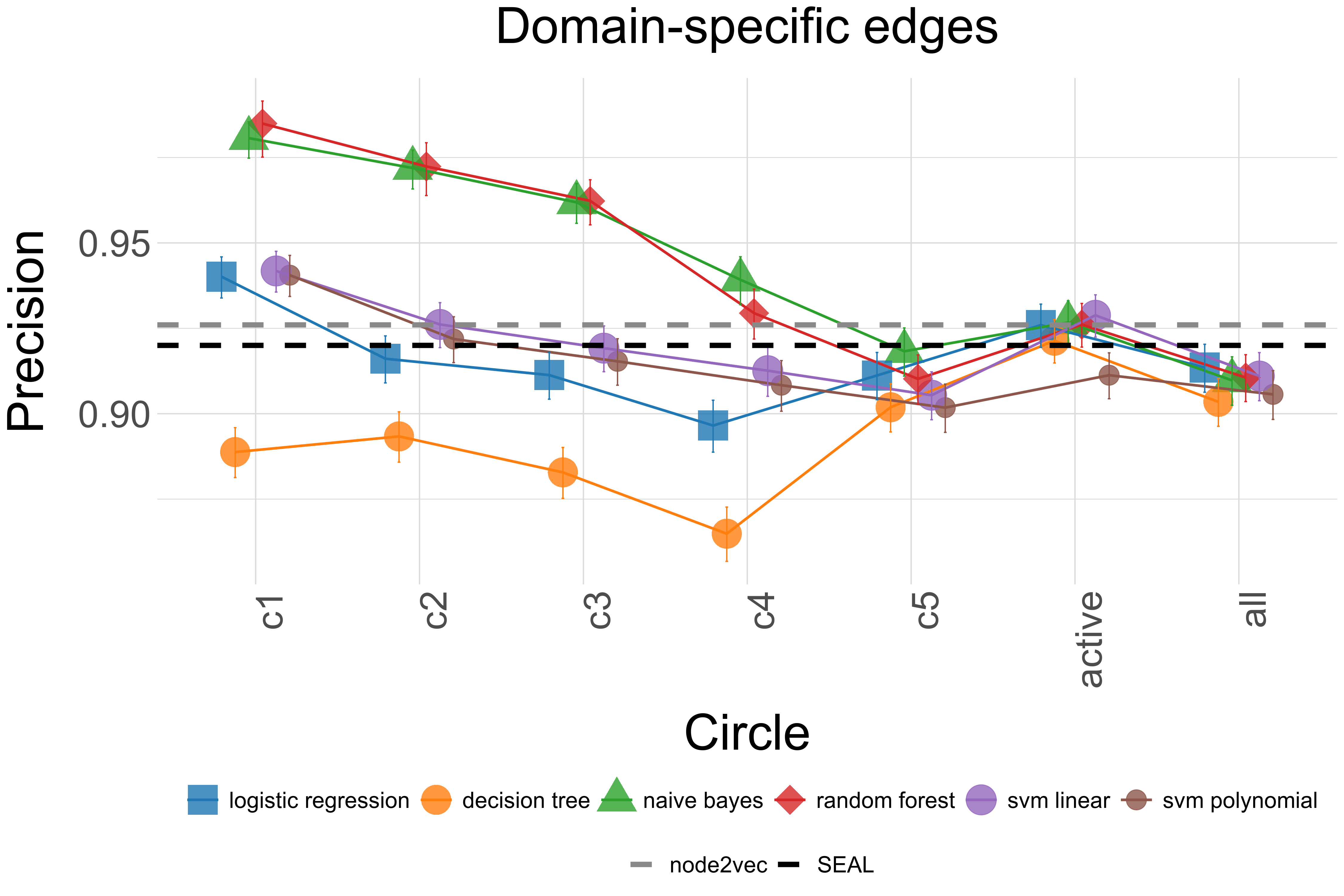}
\caption{(Gaming-related dataset) Precision (with credible intervals) - \textsc{Supervised DomainEdges} on the undersampled dataset. The $\textsc{All}$ scenario of feature learning algorithms is also included (it will be discussed in Section~\ref{sec:networkembeddings}).}
  \label{fig:supervised_precision_undersampled_domainspecific_liga} \vspace{-10pt}
\end{figure}

%
%
%

Figures~\ref{fig:supervised_precision_undersampled_alledges_a}, \ref{fig:supervised_precision_undersampled_alledges_b},  and~\ref{fig:supervised_precision_undersampled_domainspecific_liga} show the results, for the \textsc{AllEdges} and the \textsc{DomainEdges} case respectively.
Recall that, as explained in Section~\ref{sec:train_test}, these results are obtained with the negatives undersampled to make the positive and negative classes balanced. 
In Appendix~\ref{app:supervised_additional}\iftoggle{ONEFILE}{}{ of the SM~\cite{toprak2021harnessing-power}} we also investigate the link prediction performance when no undersampling is carried out (the findings are substantially confirmed). Figure~\ref{fig:supervised_precision_undersampled_alledges_a} shows the precision of the supervised link prediction approaches in the \textsc{AllEdges} scenario of gaming-related dataset. We can identify two important results.
First, all approaches that rely on C1 information provide better precision than with any other layer.
In particular, using all edges (case \textsc{All}) is never better than using C1 alone.
Second, all circle-based approaches that leverage information on the outermost social layers perform at least as well as the baseline (\textsc{All}).
When considering the generic users dataset (Figure~\ref{fig:supervised_precision_undersampled_alledges_b}), we observe that
the peak performance in C1 is confirmed for Random Forest and na\"ive Bayes. However, for the other learning strategies, the precision drastically drops in C1. For generic users, we can conclude that all the social circles but the first one provide approximately the same precision as the baseline \textsc{All}.
Finally, please note that, for both datasets, the precision achieved with supervised strategies is much higher than that obtained with unsupervised ones (consistently with the results in the related literature~\cite{Wang2011}).

In the \textsc{DomainSpecificEdges} scenario for the gaming-related dataset  (Figure~\ref{fig:supervised_precision_undersampled_domainspecific_liga}), again na\"ive Bayes and Random Forest achieve the highest precision when leveraging only C1 relationships. Differently from the previous case, the precision of decision trees in C1 is lower than with other layers. The remaining learning strategies still gain from using C1 information, but the marging is smaller in this case. Differently from the unsupervised case, supervised learning seems able to overcome the difficulties of predicting new links in the \textsc{AllEdges} scenario, hence the relative advantage of filtering out links not domain-specific is partially lost.

In summary, supervised learning consistently yields better predictions than their unsupervised counterparts.
This may be due to more flexible identification of boundaries between positives and negatives (with respect to the simple threshold-based approach of similarity-based heuristics) or to a smart combination of the similarity-based heuristics.
Regardless, social circles again prove very effective in boosting the precision of link prediction strategies, especially in the gaming-related dataset.
In particular, supervised strategies seem able to effectively exploit the innermost layer C1 in the domain-dependent dataset and the active network layer in domain-independent dataset much better than the unsupervised cases.

\vspace{-10pt}
\subsection{Comparison against link prediction methods based on feature learning}
\label{sec:networkembeddings}

Until now, we have studied the social-aware link prediction method by embedding the intimacy levels of Dunbar's ego network model into existing unsupervised and supervised approaches. The ego network model is a way of capturing the local information of relationships of the individuals. Since the aim of the study is to understand the contribution of the social circles to the performance of the link prediction methods, we selected prediction methods where the feature calculation is not a black-box, and the link prediction method is not an end-to-end method. All the methods discussed so far were based on explicit graph features. In this section, we compare the prediction performance of the methods based on latent features -- \emph{node2vec}, leveraging shallow node embeddings, and \emph{SEAL}, which is GNN-based (see discussion in Section~\ref{sec:latent_features_based}) -- to feature extraction methods leveraging social circle information.

Table~\ref{tab:selfsupervised_results} shows the precision results of \emph{node2vec} (for the best $(p, q)$ pair), SEAL (whose best precision is achieved with $h=1$ without embeddings), and the best supervised social-aware algorithms (whose the best results always occur with C1) on the undersampled\footnote{For consistency with the original \emph{node2vec} and SEAL papers, where the negatives are undersampled in the evaluation.} generic users and gamers datasets.
Recall that we are using \emph{node2vec} and \emph{SEAL} approaches as baseline methods where all edges are included without any slicing (this corresponds to the case \textsc{All} in the previous section). Please note that this comparison is fair, since we are studying machine learning-based algorithms trained on explicit features against a machine learning-based algorithm trained on latent features. Vice versa, we do not consider the unsupervised approaches from Section~\ref{sec:results_unsupervised} as they do not leverage machine learning techniques. Table~\ref{tab:selfsupervised_results} shows that the best circle-aware link prediction methods always outperform latent feature--based algorithms. 
Looking at Figures~\ref{fig:supervised_precision_undersampled_alledges_a}, \ref{fig:supervised_precision_undersampled_alledges_b},  and~\ref{fig:supervised_precision_undersampled_domainspecific_liga} for a circle-by-circle comparison, we observe that either or both feature learning strategies always provide the best precision in the \textsc{All} case. As expected, feature learning strategies are also generally very competitive in the other cases, outperforming circle-aware approaches for several slicing. However, they are never able to surpass the precision achieved by the best circle-aware strategy.

\ul{The implications of the above results are far-reaching: by using knowledge on a few common strong ties, circle-aware link prediction consistently beats black-box approaches.} This advantage is not even paid in terms of computational complexity: as we show in Appendix~\ref{app:complexity_analysis}\iftoggle{ONEFILE}{}{ of the SM~\cite{toprak2021harnessing-power}}, all the strategies in Table~\ref{tab:selfsupervised_results} are linear in the number of edges.

\begin{table}[t]
\centering
\caption{Precision scores of \emph{node2vec} with best $(p,q)$ combinations, \emph{SEAL} with $h=1$ without embeddings (always the best configuration in our scenarios), and the best circle-aware supervised methods (all use C1 only, the learning method is \emph{RF}-Random Forest, \emph{NB}-na\"ive Bayes). The bold cells represent the best precision score per column/scenario.}
\footnotesize
\begin{adjustbox}{width=0.48\textwidth}
\small
\begin{tabular}{@{}llll@{}} \toprule
 & \textbf{Generic Users} & \multicolumn{2}{l}{\textbf{Gamers}}\\\cmidrule{3-4}
 &  & \textsc{AllEdges} & \textsc{DomainEdges}\\ \midrule
\emph{node2vec} & \vtop{\hbox{\strut 0.955}\hbox{\strut $(p=1, q=4)$}} & \vtop{\hbox{\strut 0.894}\hbox{\strut $(p=0.25, q=0.5)$}} & \vtop{\hbox{\strut 0.926}\hbox{\strut $(p=0.25, q=4)$}} \\ \cmidrule{1-4}
\emph{SEAL} & 0.920 & 0.931 & 0.920 \\ \cmidrule{1-4}
\emph{circle-aware link prediction} & \vtop{\hbox{\strut \textbf{0.974}}\hbox{\strut \emph{(C1; RF \& NB)}}} & \vtop{\hbox{\strut \textbf{0.945}}\hbox{\strut \emph{(C1; NB)}}} & \vtop{\hbox{\strut \textbf{0.985}}\hbox{\strut \emph{(C1; RF)}}} \\
\hline \end{tabular}
\end{adjustbox}
\label{tab:selfsupervised_results}
\vspace{-15pt}
\end{table}
 


\vspace{-10pt}
\section{Conclusion}
\label{sec:conclusion}

In this paper, we have studied the performance of circle-aware feature-extraction link prediction algorithms. Specifically, relying on very well-established models from anthropology, we have considered the social circles in individual ego networks, using the circle as a proxy of intimacy. We have selected four benchmark heuristics  and we have modified them to include awareness of the social circles. 
%
%
Our results show that social-circles-based link prediction is generally extremely effective. Specifically, in the vast majority of cases, regardless of the prediction approach (unsupervised or supervised), the specific heuristic or learning algorithm, and the metric (precision, AUC, F1 score) considered, leveraging social circles information outperforms the corresponding baseline in which circles are ignored. 
In addition, using only information about the innermost social circles guarantees the same performance achieved when using the whole network. 
Using social circles information also seems to provide the same performance as using additional classifiers on nodes, which might be impractical or costly to set up. Finally, and most importantly, circle-aware supervised link prediction outperformed recent state-of-the-art feature learning-link prediction approaches, including a GNN-based solution. Interestingly, the best-performing circle is C1, which comprises only the two or three strongest relationships of the ego: by using knowledge on a few common strong ties, circle-aware link prediction consistently beats black-box approaches. 

As future work, we plan to investigate in detail the reasons why different social circles play a different role for link predictions, linking them to macroscopic or microscopic topological properties of the social graph. Another promising future direction is to study whether dynamic ego networks (i.e., time-varying ego network models) may prove useful for link prediction in dynamic social graphs.

\vspace{-10pt}



\ifCLASSOPTIONcaptionsoff
  \newpage
\fi



%
\bibliographystyle{IEEEtran}
\bibliography{bibliography}

\iftoggle{ONEFILE}{
	\appendices

\renewcommand\thefigure{A.\arabic{figure}} 
\setcounter{figure}{0}

\renewcommand\thetable{A.\arabic{table}} 
\setcounter{table}{0}

\section{Data collection}
\label{app:datasets}

In order to collect an initial community of users interested in indie games, we first had to identify some relevant indie games. To this purpose, we referred to Steam, a digital game distribution platform (\url{https://store.steampowered.com/}), and SteamSpy (\url{https://steamspy.com/}) , which provides statistics about the games on Steam. Among the indie games listed as the most popular during October 2017, we were able to identify the Twitter accounts of 133 of them.
We then downloaded the timelines of these games using the Twitter Search API. From the downloaded timelines, the 400 most frequently used hashtags have been extracted and used to monitor Twitter using the Twitter Streaming API. This allowed us to identify a set of 8,932 users engaging in game-related conversations. We labelled these users as \emph{gamers}. We have then collected (June 2018) their timeline (most recent 3200 tweets) using the Twitter Search API.
In order to enrich the set of game-related nodes in the dataset, we also downloaded all followers of the initial 133 games, thus bringing in additional 25,014 nodes, for a total of 31,091 domain-specific nodes. Note that we do not collect the ego networks of these nodes.

\section{Feauture learning methods used for benchmarks}
\label{app:feature_learning}

\emph{node2vec} is a framework that maps the nodes of a graph to low-dimensional feature vectors while preserving nodes' neighborhoods \cite{grover2016node2vec}. It proposes a biased random walk based on breadth-first sampling (BFS) and depth-first sampling (DFS) with transition probabilities $p$ and $q$. BFS  captures the embeddings related to structural equivalence (e.g., being a hub node) while DFS  captures embeddings related to communities based on homophily. There exist other four parameters for \emph{node2vec}: $d$ the dimension/length of the feature vector, $l$ the length of random walk, $k$ the size of the neighborhood, and $r$ the number of walks per node. A feature vector is produced per node in the network by combining the Skip-gram architecture of \emph{word2vec} with a flexible neighborhood definition obtained with biased random walks. For the link prediction task, these node embeddings are transformed into edge embeddings by applying binary operators like, e.g., Hadamard. Then, logistic regression is applied on the edge embeddings to train a binary classification model on positive and negative data (existing and missing links). 
Given that \emph{node2vec} is designed to capture latent features, in this study, we apply \emph{node2vec} method as is, without explicitly adding information on the social circles. 
Since the aim is to predict the possible future links between the egos whose social circles can be computed, as described in Section~\ref{sec:prediction_algorithm}, we learn the node embeddings of the vertex set $V_{e}$ by using all edges $E_{e}$. For the extraction of node embeddings, we try the different combinations of $p$ and $q$ values ($p,q \in \{0.25, 0.5, 1, 2, 4\}$) while keeping other parameters of random walk fixed ($d=128$, $l=80$, $k=10$, $r=10$) as mentioned in \cite{grover2016node2vec}. Then, we calculate the edge embeddings of all existing and missing links in $V_e \times V_e$. Given two feature vectors, $f(i)$ and $f(j)$, of egos $i$ and $j$ where $i,j \in V_e$, the edge embedding of the edge $e_{ij}$ is calculated as: $f(i) \circ f(j)$ where $\circ$ is a binary operator (average, Hadamard, weighted-L1, weighted-L2). We are using Hadamard operator, which is shown to be the most effective one in \cite{grover2016node2vec}.

\emph{SEAL}~\cite{Zhang2018} is a graph neural network method that focuses on the local enclosing subgraphs of the target links (node pairs). For a node pair $i,j$, the local enclosing subgraph is the subgraph induced by the union of $i$'s and $j$'s neighborhoods, up to $h$ hops. Once the enclosing subgraphs are extracted, the nodes in them are labelled using the Double-Radius Node Labeling (DRNL), whereby the labels are a function of the distance from the target nodes $i,j$. Parameter $h$ defines the order of the neighborhood. Zhang and Chen~\cite{Zhang2018} show that higher-order graph structures and their corresponding heuristics (such as Page Rank, Katz index) can be captured using small $h$ values. Here, we use all relationships ($E$) to extract the local enclosing subgraph of all ego node pairs on which we apply the prediction task. We also set $h \in \{1, 2\}$, like in~\cite{Zhang2018}.
In SEAL, the extracted local enclosing subgraphs, for the existing and missing links, are then fed to a Deep Graph Convolutional Neural Network (DGCNN)~\cite{zhang2018end} to train the model. For the DGCNN, we use the same parameters of the SEAL paper. The core structure of \emph{SEAL} consists of local enclosing subgraphs and DGCNN, which can be extended with node features (which we don't have) and node embeddings (such as those of \emph{node2vec}). In~\cite{Zhang2018}, it was shown that node embeddings may o may not increase the prediction performance. Thus, here we carry out experiments both with and without embeddings.

\section{Preliminary ego network analysis}
\label{sec:results_egonets}

As anticipated in Section~\ref{sec:dataset}, we extract the ego network structure of the gamers and generic users in $V_e$ using the methodology described in~\cite{boldrini2018twitter}. The resulting distribution of the optimal circle number of gamer egos and generic user egos can be seen in Figure~\ref{fig:optimal_circle_distr}. The mean value is 4.94 for gamers and 4.73 for generic users, and the median and mode values are equal to 5 for both datasets. The plot shows that the number of social circles of both types of egos is compatible with the findings of previous studies which establish that human ego network can be layered into 5 layers/circles. For further understanding of the ego network structures, we examine the alter distribution through the layers. In previous studies~\cite{Dunbar2015}, it has been shown that in OSN the external layers are slightly smaller than the reference model discussed in Section~\ref{sec:relwork_egonetworks}. In Figure~\ref{fig:c5_alter_distr}, we show the distribution of the alter number for the egos with optimal circle number equal to~5 (the case generally studied in the related literature). These results, both for gamers and generic users, are compatible with the typical findings from OSNs analysis in the related literature.

\begin{figure}[t]
\subfloat[Gaming-related dataset ($V_{e}$)
\label{fig:optimal_circle_distr_a}]
{\includegraphics[width=0.49\textwidth]
{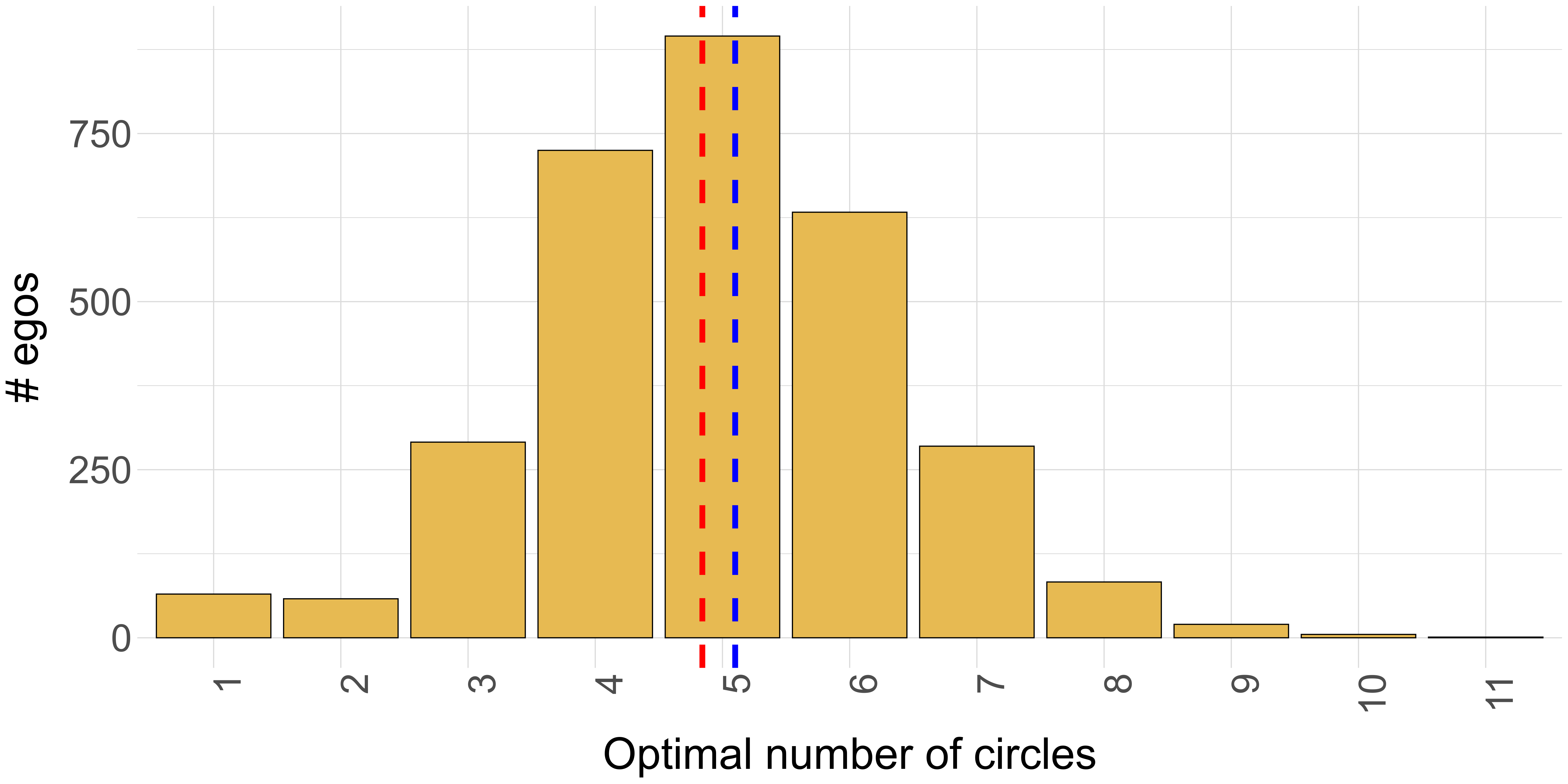}}
\hfill
\subfloat[Generic users dataset ($V_{e}$)
\label{fig:optimal_circle_distr_b}]
{\includegraphics[width=0.49\textwidth]
{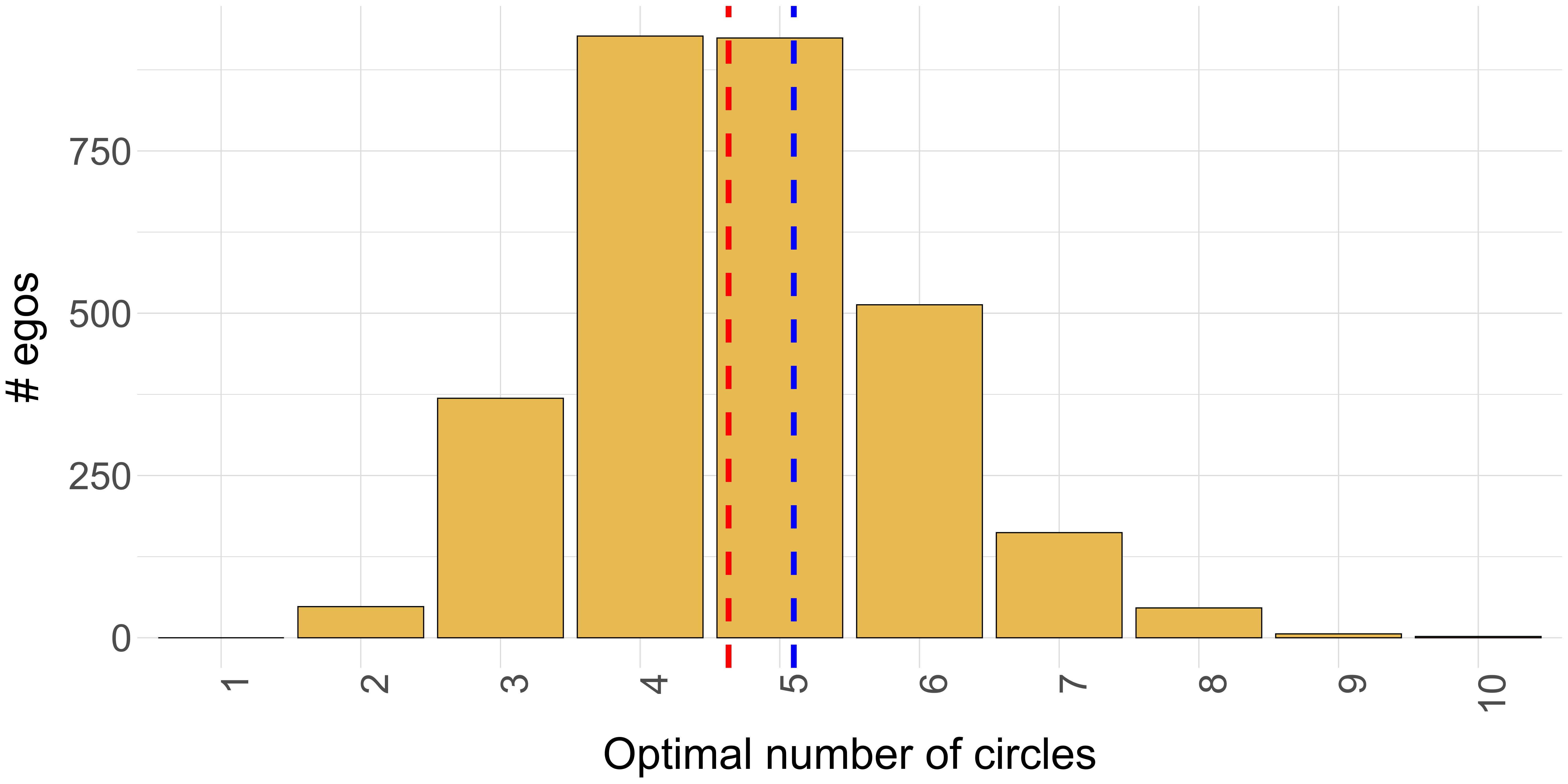}}
\hfill
\caption{Distribution of the optimal number of social circles for the gamer egos and generic user egos. Red and blue lines represent the mean and median values, respectively.}
\label{fig:optimal_circle_distr}
\end{figure}

\begin{figure}[t]
\subfloat[Gaming-related dataset ($V_{e}$)
\label{fig:c5_alter_distr_a}]
{\includegraphics[width=0.49\textwidth]
{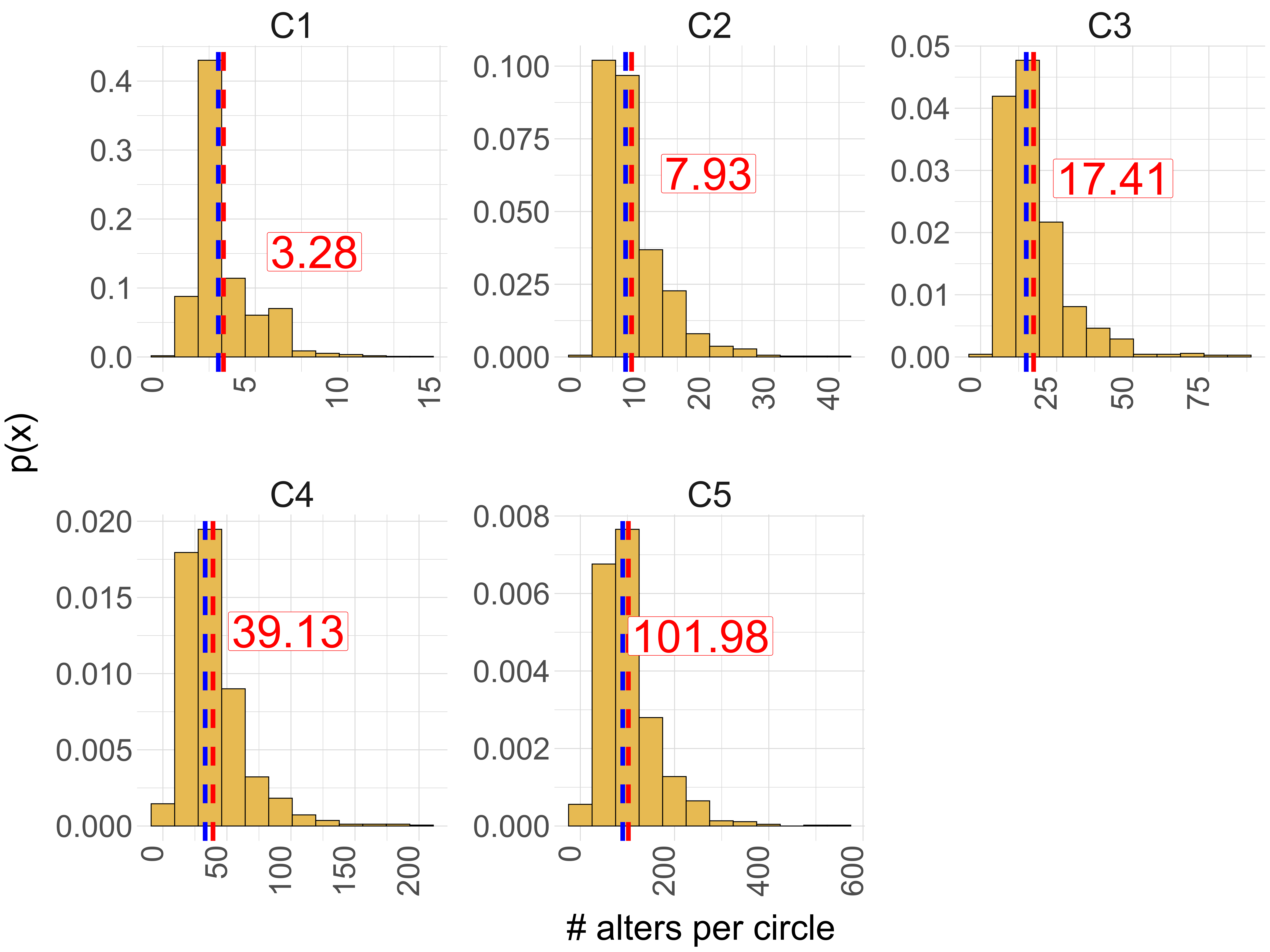}}
\hfill
\subfloat[Generic users dataset ($V_{e}$)
\label{fig:c5_alter_distr_b}]
{\includegraphics[width=0.49\textwidth]
{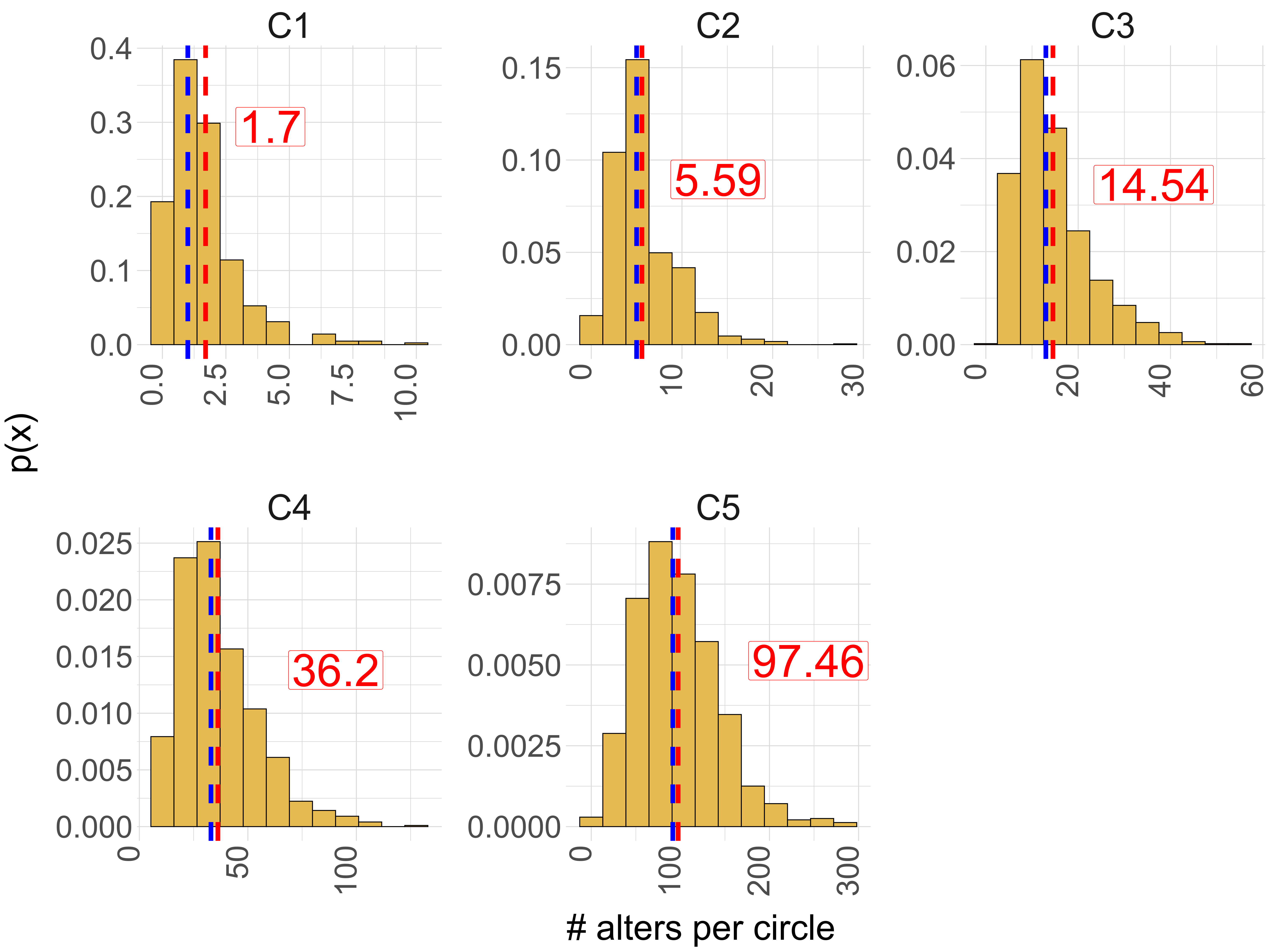}}
\hfill
\caption{Average alter distribution per each circle of gamer egos and generic user egos. For the sake of clarify, egos that have optimal circle number equal to five have been used. Red and blue lines represent the mean and median values, respectively, of the distribution.}
\label{fig:c5_alter_distr}
\end{figure}

\section{Evaluation setup}
\label{app:evaluation_setup}

Traditionally, binary classifiers are evaluated using the ROC curve, which plots the True Positive Rate ($\frac{TP}{TP + FN}$) against the False  Positive Rate ($\frac{FP}{FP + TN}$). However, the scenario we are considering clearly suffers from the class imbalance problem, since the actual edges between nodes are significantly fewer than all the possible edges that there could exist between them. 
In this situation, the ROC curve, while being insensitive to the class imbalance\footnote{Differently from the accuracy, which should not be used when imbalance is an issue. In fact, given the accuracy formula $\frac{TP+TN}{P+N}$, it is easy to see that correct classification in the most numerous class will hide anything about the sparse class.}, tends to provide overoptimistic results~\cite{Garcia-Gasulla2015}.
The related literature on evaluating systems with class imbalance suggests using metrics like precision ($=\frac{TP}{TP + FP}$), recall (aka TPR) and $F_{1}$ score. As already argued by~\cite{Wang2011}, in the context of link prediction, precision is more important than the other metrics (such as recall), because if the precision is high, one can live with some false negatives. Therefore, the focus of this evaluation will be mostly on precision. For the completeness of results, we also provide, when relevant, the F1 score and the AUC of the precision-recall curve. Note also that we are interested in comparing the performance obtained using the baseline approach (no social circles information, all edges are considered) and the other strategies in which only edges belonging to the ego network circles are considered.

Precision and recall are evaluated based on metrics (TP, FP, TN, FN) which corresponds to specific realizations of a target phenomenon. 
Thus, the question arises of how confident we can be about the results obtained on a particular collection that is the result of random sampling (as in the case of cross-validation or subsampling of negatives for supervised prediction). 
\cite{Goutte2005} proposes a Bayesian approach to the problem. 
The basic idea is to model TP, FP, FN, TN as if they were sampled from a multinomial distribution, with unknown success probability per category $\pi_i$. 
Based on the properties of the multinomial distribution, we know that individual counts are binomial (i.e., equivalent to the number of successes in a sequence of $n$ independent experiments, each with success probability $p$). 
Specifically, it can be proved that TP are binomial with parameters TP+FP (corresponding to the number of trials) and success probability $p$ (corresponding to the precision). 
Using Bayes' theorem and well-known priors associated with binomials, it can be derived that the posterior of the precision~$p$ is Beta distributed:
\begin{equation}
p | \mathcal{D} \textrm{ (aka the posterior distr.)} \sim Beta(TP+\lambda, FP+\lambda).
\end{equation}
A similar formula is obtained for the recall $r$. 
In these formulas, $\lambda$ is a parameter whose value depends on the prior considered: $\lambda=\frac{1}{2}$ corresponds to Jeffrey's non-informative prior (the one we use here), $\lambda=1$ to the uniform prior.
%
Credible intervals\footnote{Using confidence intervals instead of credible intervals would not be the best choice in this case, because precision, recall, and F1 score belong to $[0,1]$, hence, especially for values close to either 0 or 1, they depart significantly from the normality assumption behind confidence intervals.} for $p$ and $r$ can then be obtained as the intervals containing 95\% of the Beta  distribution. The formula for the credible interval of the F1 score is slightly more complicated, so we do not report it here.
Microaveraging and credible intervals are also used in~\cite{Wang2016}.

\section{Evaluation of supervised link prediction - additional results}
\label{app:supervised_additional}

As already mentioned, the precision achieved with supervised learning is generally much higher than that observed with unsupervised algorithms. 
In order to investigate whether this might be due to an artificial bias introduced by the undersampling of negatives, we have also run the same link prediction experiment using the complete set of negatives (again splitting 90\%-10\% between train and test). 
Figures~\ref{fig:supervised_precision_all_alledges_a}, \ref{fig:supervised_precision_all_alledges_b},  and~\ref{fig:supervised_precision_all_domainspecific_liga} show the precision obtained in this case. 
Only logistic regression, decision trees, and na\"ive Bayes are shown, since the others did not converge in a reasonable time due to to new size of the problem. 
The precision values achieved in this case are significantly smaller than those obtained with the undersampled dataset (thus confirming the limitations of this approach). 
Still, they improve over the ones obtained with the standalone similarity measures.
Again, it is confirmed that the C1 layer is the most useful for making precise predictions in the gaming-related dataset (Figures~\ref{fig:supervised_precision_all_alledges_a}). 
This result is very encouraging, since C1 is the most lightweight in terms of resources consumed. 
In the generic users dataset (Figures~\ref{fig:supervised_precision_all_alledges_b}), the role played by the social circles is reversed: using C5 or \textsc{Active}, we are able to perform predictions that are at least as good as the baseline, however the precision drops as we move towards the innermost layers. 

\begin{figure}
\subfloat[Gaming-related dataset
\label{fig:supervised_precision_all_alledges_a}]
{\includegraphics[width=0.49\textwidth]
{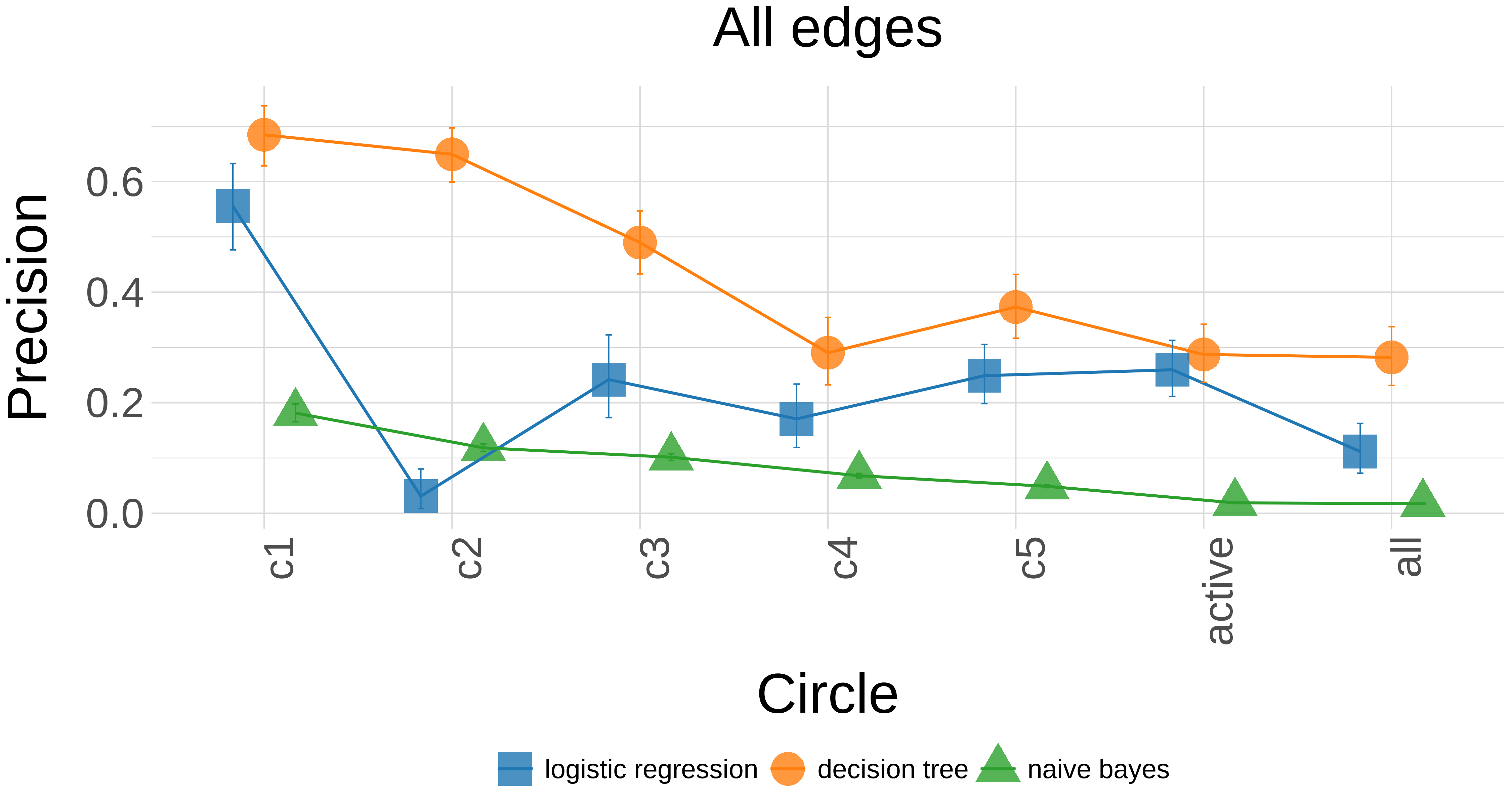}}
\hfill
\subfloat[Generic users dataset.
\label{fig:supervised_precision_all_alledges_b}]
{\includegraphics[width=0.49\textwidth]
{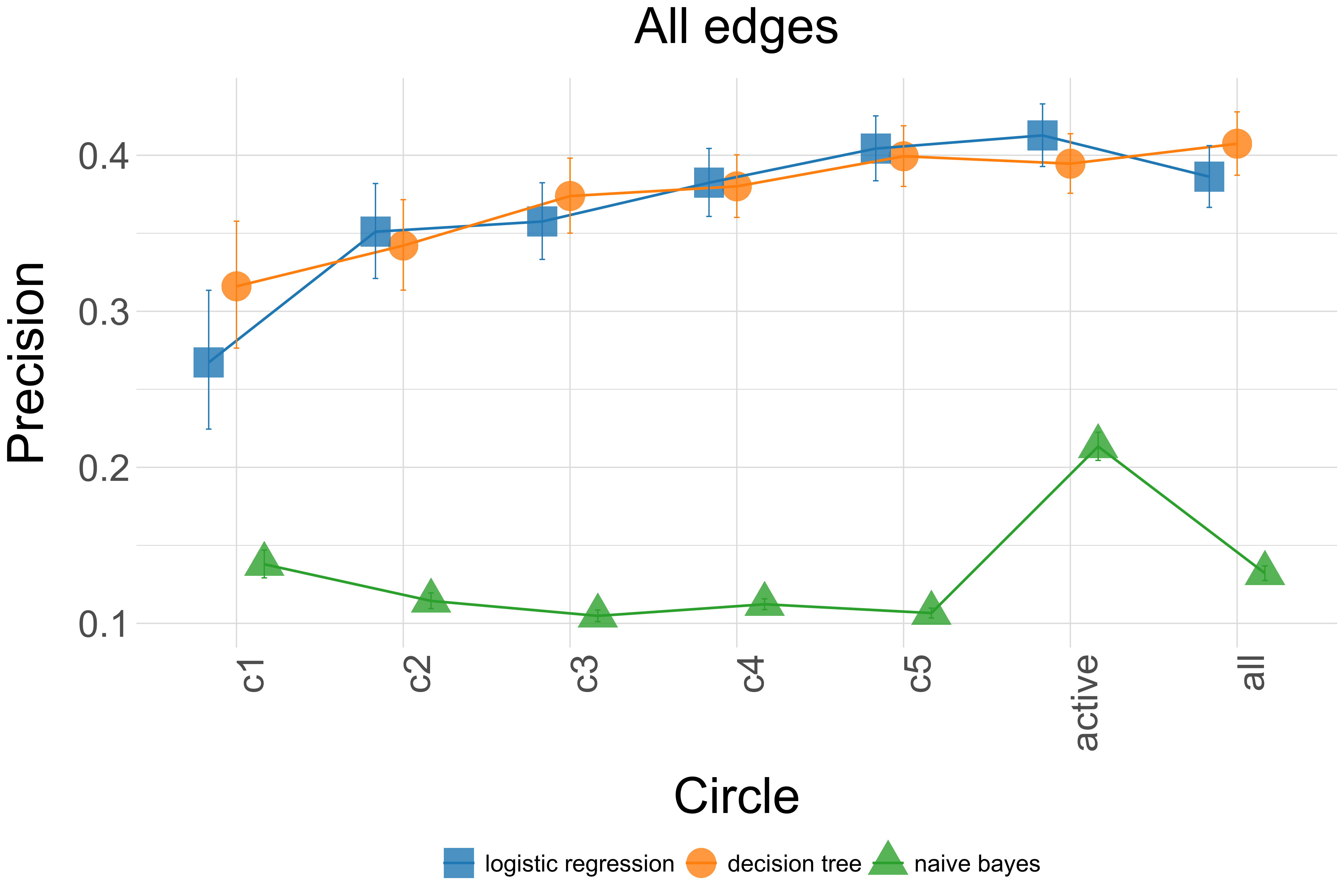}}
\hfill
\caption{Precision (with credible intervals) - \textsc{Supervised AllEdges} on the full datasets.}
\label{fig:supervised_precision_all_alledges}
\end{figure}

\begin{figure}[t]
\centering
  \includegraphics[width=0.49\textwidth]{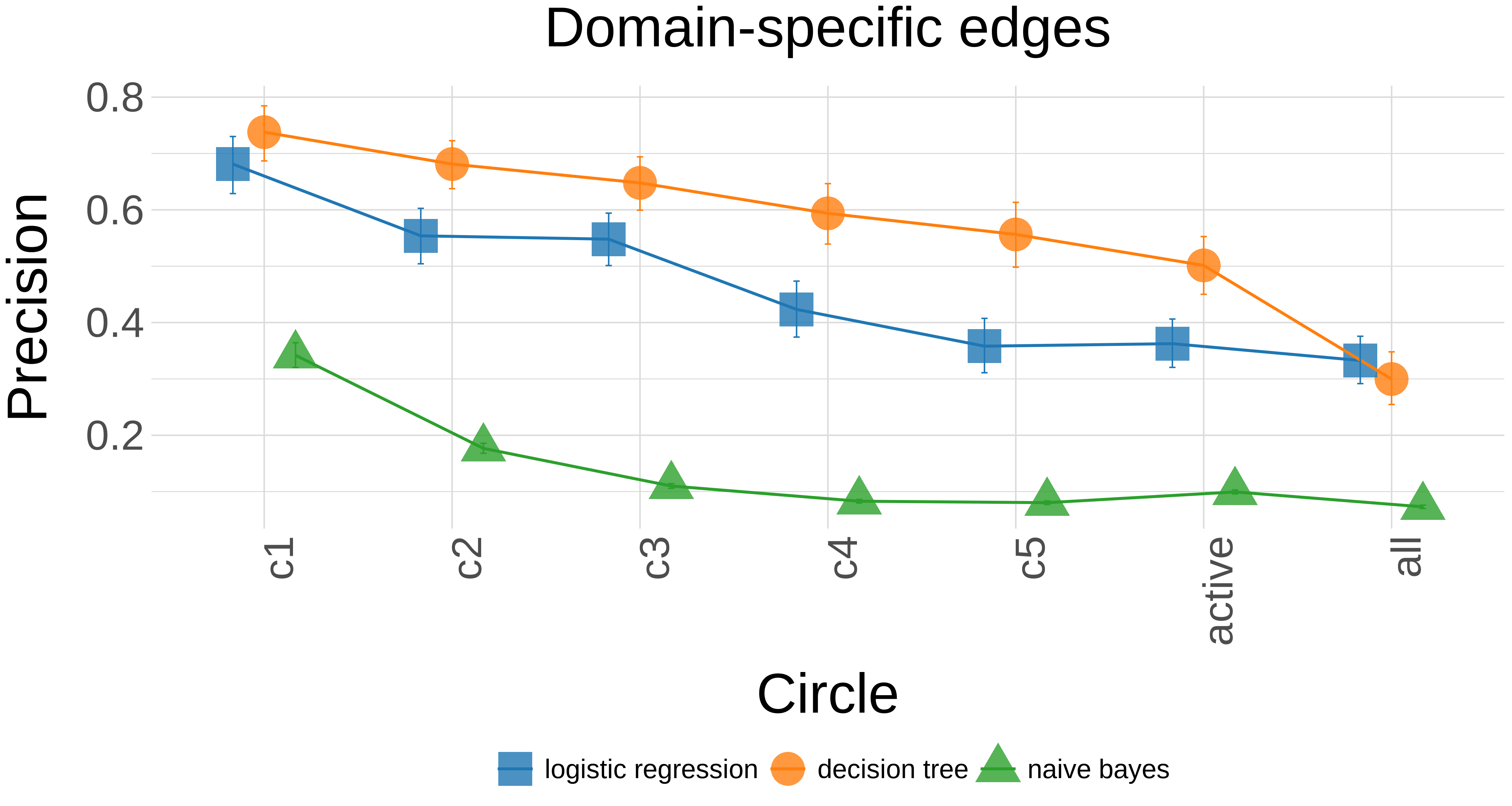}
  \caption{Precision (with credible intervals) - \textsc{Supervised DomainEdges} on the full gaming-related dataset.}
  \label{fig:supervised_precision_all_domainspecific_liga}
\end{figure}

We conclude the section analysing the F1 score, which is the harmonic mean of precision and recall. We only show the results obtained using the full negative set, as they provide a more reliable evaluation. For both the \textsc{AllEdges} and the \textsc{DomainEdges} case (Figures~\ref{fig:supervised_fscore_all_alledges_a}, \ref{fig:supervised_fscore_all_alledges_b}, and ~\ref{fig:supervised_fscore_all_domainedges_liga}), we observe a different ranking with respect to the corresponding precision results, with the highest F1 being achieved by the na\"ive Bayes approach. 
Still, since, as discussed in~\cite{Wang2016}, precision is much more important than recall in link prediction, we argue that decision trees should be preferred to na\"ive Bayes. 
If we only look at the role played by the social circles, we observe that leveraging social circles is always better or as good as the \textsc{All} baselines. However, when aiming at striking a balance between precision and recall, C1 is replaced by C2 as the best performing circle in the gaming-related dataset, probably due to C1's high selectivity for the most intimate relationships. In the generic users dataset, instead, the \textsc{Active} layer is confirmed as the one providing the best prediction performance.

\begin{figure}
\subfloat[Gaming-related dataset
\label{fig:supervised_fscore_all_alledges_a}]
{\includegraphics[width=0.49\textwidth]
{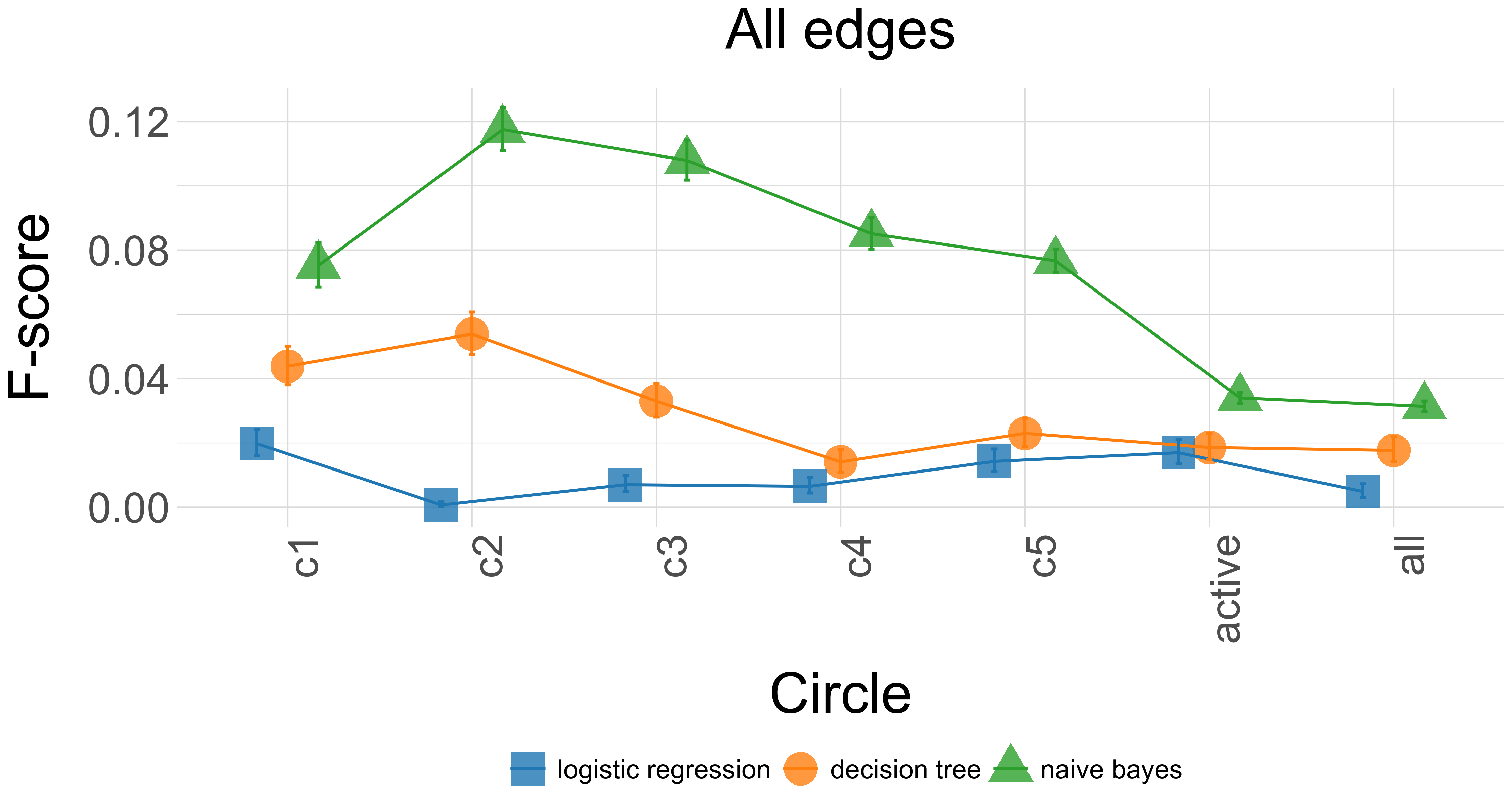}}
\hfill
\subfloat[Generic users dataset.
\label{fig:supervised_fscore_all_alledges_b}]
{\includegraphics[width=0.49\textwidth]
{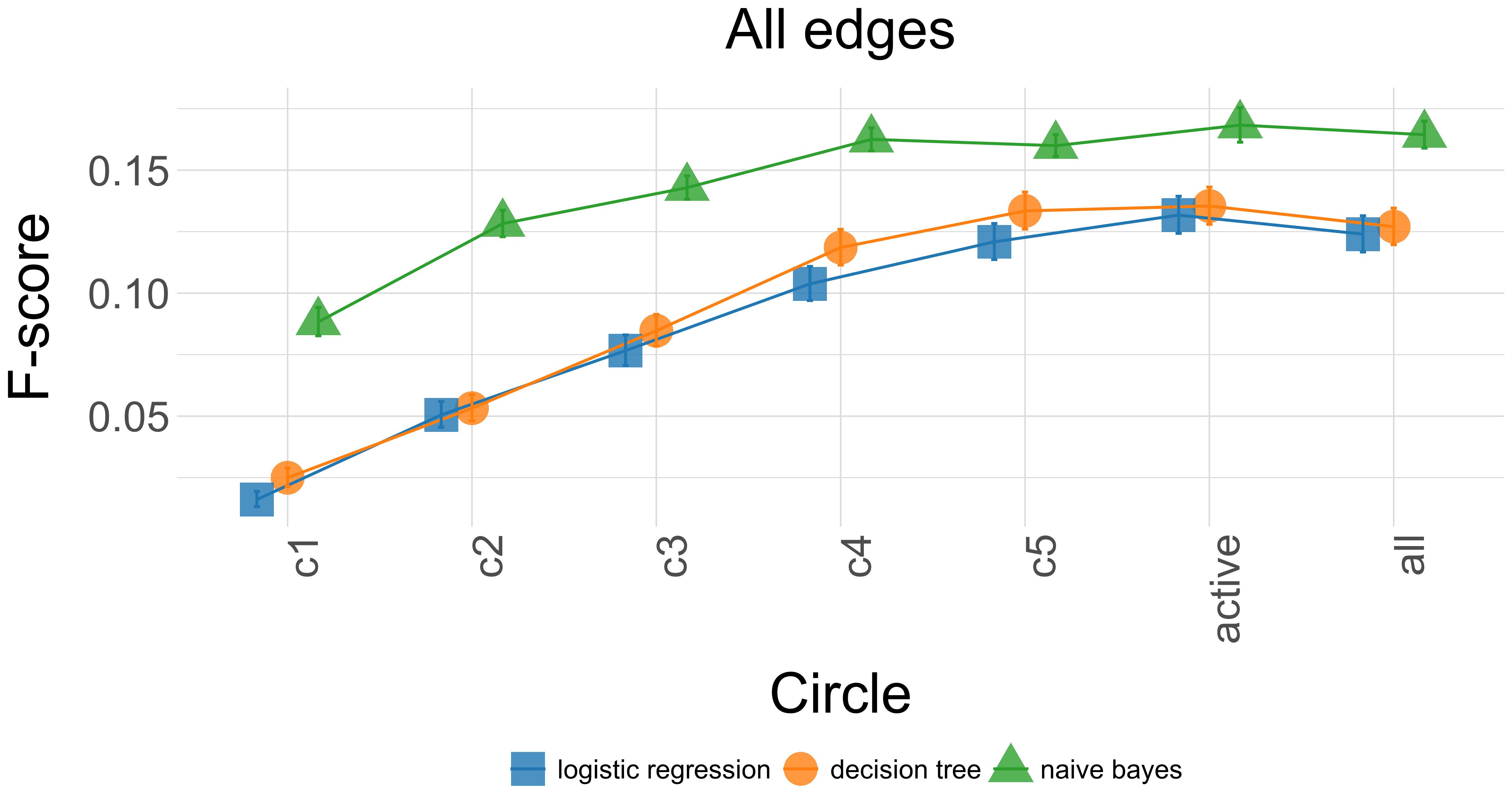}}
\hfill
\caption{F1 score (with credible intervals) - \textsc{Supervised AllEdges} on the full datasets.}
\label{fig:supervised_fscore_all_alledges}
\end{figure}

\begin{figure}[t]
\centering
  \includegraphics[width=0.49\textwidth]{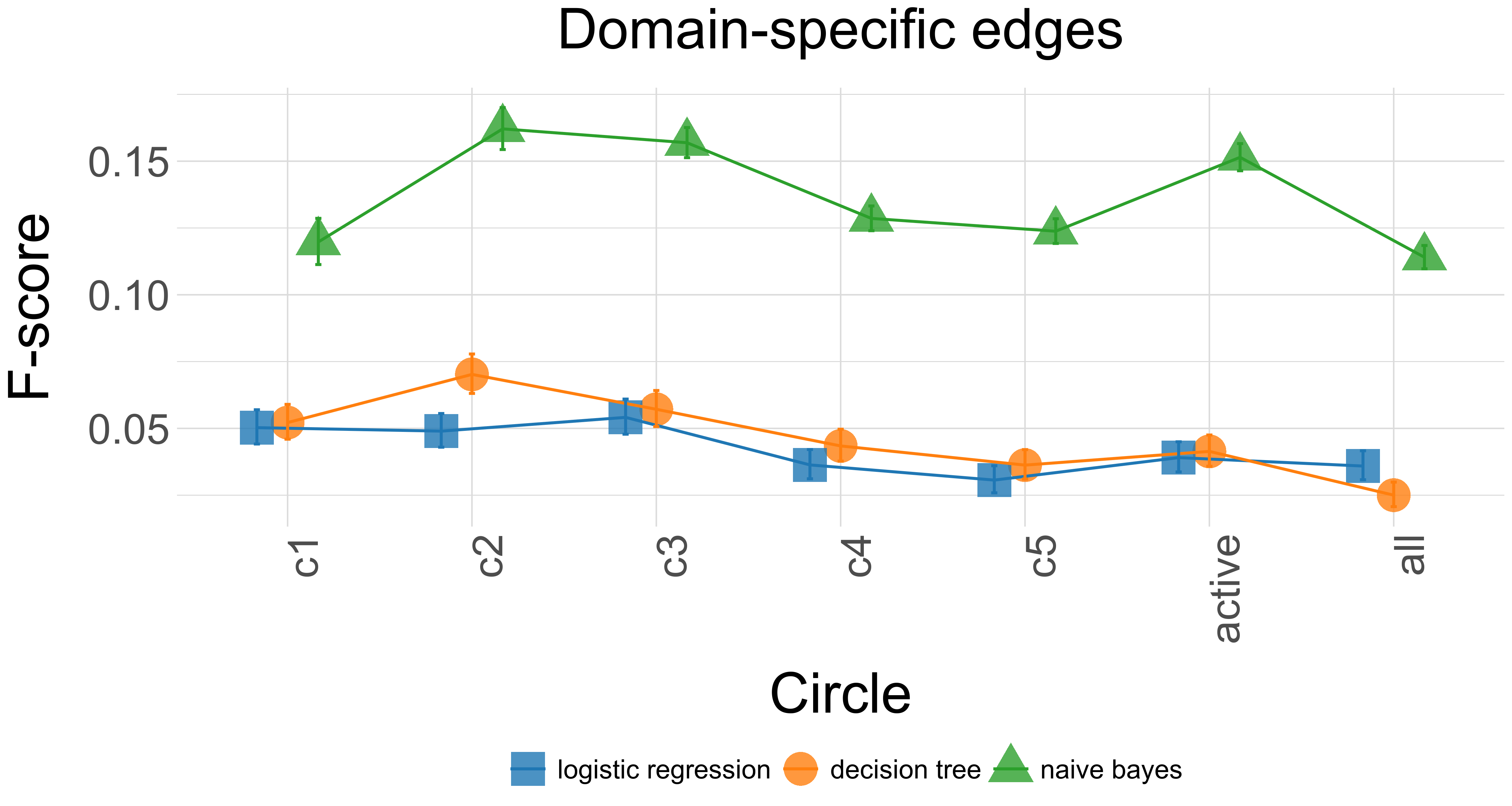}
  \caption{F1 score (with credible intervals) - \textsc{Supervised DomainEdges} on the full gaming-related dataset.}
  \label{fig:supervised_fscore_all_domainedges_liga} \vspace{-10pt}
\end{figure}

\section{Complexity analysis}
\label{app:complexity_analysis}


In this section, we investigate the computational complexity of ego-aware link prediction and how it compares to the other approaches tested in Section~\ref{sec:evaluation}.
%
For the convenience of the reader, the notation used in this section is summarised in Table~\ref{tab:complexity_notation}.
%

\begin{table}[h]
\caption{Notation summary for complexity analysis}
\begin{adjustbox}{width=0.48\textwidth}
\begin{tabular}{@{}ll@{}}
\toprule
\textbf{Symbol} & \textbf{Description} \\
\midrule
$V_e$ & set of ego nodes \\
$E_e$ & set of existing edges between ego nodes \\
$\overline{E}_e$ & Set of non-existing edges between ego nodes \\
$E_{train}$ & Set of training edges, it is the union set of positive and negative edges used for training \\
$\Gamma(i)$ & Neighborhood of a generic node $i$ \\
$\Gamma_{C_x}(i)$ & Social circle $C_x$ of a generic node $i$ \\
$d$ & Number of features used for learning \\
$h$ & Depth of the enclosing subgraph computed by SEAL \\
\bottomrule
\end{tabular}
\end{adjustbox}
\label{tab:complexity_notation}
\end{table}

\begin{algorithm}[t]
\DontPrintSemicolon
\KwData{$\mathcal{G}=(V_e,E)$; $C_x$ = the social circle to be considered.}
\SetKwInput{KwData}{Notation}
\KwData{$\Gamma(i) =$ neighborhood of node $i$; $\overline{E}_e =$ set of missing edges.}
\KwResult{Similarity $sim_{C_x}(i,j)$ for each node pair $i,j$ in $V_e$.}
\Begin{
\For{$i\in V_e$}{
\tcp{Compute the ego network and slice the graph according to $C_x$}\label{ln:cmt}
$AssignToCircles(\Gamma(i))$\; \label{ln:circles}
$\Gamma_{C_x}(i) = Slice(\Gamma(i), C_x)$\; \label{ln:slice}
}
\For{$(i,j)\in \overline{E}_e$}{ \label{ln:sim}
    \tcp{Compute the similarity}
    \If{CN}{ 
        $sim_{C_x}(i,j) = | \Gamma_{C_x}(i) \cap \Gamma_{C_x}(j)| $\; \label{ln:cn}
    }
    \If{JC}{ \label{ln:jc}
        $sim_{C_x}(i,j) = \frac{| \Gamma_{C_x}(i) \cap \Gamma_{C_x}(j) |}{| \Gamma_{C_x}(i) \cup \Gamma_{C_x}(j) |}$\;
    }
    \If{AA or RA}{ \label{ln:aa_ra}
        $sim_{C_x}(i,j) = 0 $ \;
        \For{$z \in \Gamma_{C_x}(i) \cap \Gamma_{C_x}(j)$}{
            $sim_{C_x}(i,j) = sim_{C_x}(i,j) + penalised(\Gamma(z))$\;
        }
    }
}
}
\caption{Find the similarity between nodes\label{algo:sim} - unsupervised setting}
\end{algorithm}

Let us start with heuristic-based link prediction in the unsupervised case (discussed in Section~\ref{sec:unsupervised}). Algorithm~\ref{algo:sim} illustrates the steps towards computing the similarity between a pair of nodes depending on the specific similarity considered. Please note that the weighted graph discussed in Section~\ref{sec:dataset} is assumed as input. With respect to the baselines (traditional methods which are represented as \textsc{All} in the results) in which no circle information is used, our approach requires two additional steps: (i) computing the ego networks of the nodes for which predictions should be made (line~\ref{ln:circles}) and (ii) pruning the ego networks based on the social circle we want to leverage (line~\ref{ln:slice}). The code from line~\ref{ln:sim} on is common to both the baselines and the circle-aware approaches, since the specific slicing considered is controlled with $C_x$ (and $x = \textsc{All}$ for the baseline). In Lemma~\ref{theo:time_complexity_egonet} below we derive the time complexity for computing the ego networks. Note that, in social networks where nodes are real people, the neighborhood size does not grow significantly with $|V|$, hence the complexity of computing an ego network can be considered approximately constant.

\begin{lemma}[Complexity egonets]\label{theo:time_complexity_egonet}
The complexity of computing the ego network of a generic node $i$ is given by the following:
\begin{equation}\label{eq:complexity_ego}
 \mathcal{O}(\textrm{ego}) = \mathcal{O}(|\Gamma(i)| * \log(|\Gamma(i)|).
\end{equation}
\end{lemma}

\begin{proof}
Line~\ref{ln:circles} in Algorithm~\ref{algo:sim} entails a) sorting the neighbors of node $i$ ($\Gamma(i)$) based on contact frequencies, b) extracting the active part (corresponding to relationship with contact frequency larger than once per year, as discussed in Section~\ref{sec:relwork_egonetworks}), and c) running a unidimensional clustering algorithm on the remaining contact frequencies. The complexity of step (a) is $\mathcal{O}(|\Gamma(i)|*\log(|\Gamma(i)|))$ (e.g., assuming standard merge sort is used~\cite{cormen2009introduction}), the complexity of step (ii) is $\mathcal{O}(\log(|\Gamma(i)|))$ which is equivalent to the complexity of a search algorithm on a sorted input (e.g., assuming binary search is used~\cite{cormen2009introduction}). After step (b), for each ego $i$, we are not working anymore on its full neighborhood but only on the active part $\Gamma_{active}(i)$, which is smaller than $\Gamma(i)$ since it only includes strong relationships. The complexity of step (c) is mainly dependent on the chosen clustering algorithm (note that the classification of alters in circles is robust against the different clustering methods~\cite{Dunbar2015}). For example, DBSCAN is quite efficient and it runs in $\mathcal{O}\left(|\Gamma_{active}(i)| \log(|\Gamma_{active}(i)|) \right)$~\cite{tan2019introduction}. The time complexity of k-means is linear~\cite{tan2019introduction} in the number of points to cluster (i.e., $\mathcal{O}\left(|\Gamma_{active}(i)|\right)$) but it needs additional computations for the selection of the best $k$ (e.g., $\mathcal{O}(k |\Gamma_{active}(i)|)$ for the partition coefficient method, $\mathcal{O}( |\Gamma_{active}(i)|^2)$ for the Silhouette method). Mean Shift runs in $\mathcal{O}(|\Gamma_{active}(i)| \log(|\Gamma_{active}(i)|)$, like DBSCAN. The last step to carry out is the slicing in line~\ref{ln:slice}, whose complexity is equivalent to that of a search in $\Gamma_{active}(i)$, hence it is logarithmic in its size.
In summary, the complexity of computing an ego network is given by the following (assuming DBSCAN or Mean Shift are used for clustering):
\begin{align*}
 \mathcal{O}(\textrm{ego}) & = \mathcal{O}\left(|\Gamma(i)| * \log(|\Gamma(i)|) \right)+ \nonumber \\
 & \quad +  \mathcal{O}\left(\log(|\Gamma(i)|) + |\Gamma_{active}(i)| \log(|\Gamma_{active}(i)| + \right. \nonumber \\
 &  \left. \quad + \log(|\Gamma_{active}(i)|) \right) \nonumber \\ 
 & = \mathcal{O}(|\Gamma(i)| * \log(|\Gamma(i)|),
\end{align*}
where we have leveraged $|\Gamma(i)| \ge |\Gamma_{active}(i)|$. 
\end{proof}

Exploiting the above result, we can now derive the time complexity of circle-aware and baseline heuristics with unsupervised link prediction. Theorem~\ref{theo:time_complexity_unsupervised} below states that in both cases the complexity of unsupervised link prediction grows with $|V|^2$, hence it is intrinsically inefficient. Since $|\Gamma_{C_x}(i)| \ll |\Gamma(i)|$ when $x$ is small, circle-aware link prediction provides a marginal advantage in this case.

\begin{theorem}[Complexity unsupervised]\label{theo:time_complexity_unsupervised}
The time complexity of link prediction using circle-aware or baseline heuristic $\phi$ with unsupervised learning is given by the following:
\begin{eqnarray} \label{eq:complexity_social_final}
  \mathcal{O}(\textrm{circle-aware }\phi) &=& \mathcal{O}(|V|^2 * |\Gamma_{C_x}(i)|)  \\
  \mathcal{O}(\textrm{baseline } \phi) & = &\mathcal{O}(|V|^2 * |\Gamma(i)|), \label{eq:complexity_baseline_final}
\end{eqnarray}
where $\phi$ is any of the heuristics defined in Section~\ref{sec:unsupervised} (i.e., CN, JC, AA, RA), $\Gamma(i)$ denotes the neighborhood of node $i$, and $\Gamma_{C_x}(i)$ corresponds to the neighborhood of node $i$ cut at social circle $C_x$.
\end{theorem}

\begin{proof}
With reference to Algorithm~\ref{algo:sim}, in order to estimate the time complexity with circle-aware $\phi$ in the average case, we have to derive the time complexity of computing the ego network ($\mathcal{O}(\textrm{ego})$, derived in Lemma~\ref{theo:time_complexity_egonet}) and that of computing the similarity ($\mathcal{O}(\textrm{circle-aware similarity})$), as illustrated in Equation~\ref{eq:complexity_social} below. On the other hand, only the similarity complexity is needed for the baseline case (Equation~\ref{eq:complexity_baseline}). Note that, in both equations, we have approximated $|\overline{E}_e|$ (the number of missing edges between egos) as $|V|^2$, since social networks are typically sparse.
\begin{align} \label{eq:complexity_social}
  \mathcal{O}(\textrm{circle-aware } \phi) &= \mathcal{O}(|V|)\mathcal{O}(\textrm{ego}) + \nonumber\\
  & \quad + \mathcal{O}(|V|^2) *  \mathcal{O}(\textrm{circle-aware sim}) \\
   \mathcal{O}(\textrm{baseline } \phi) &= \mathcal{O}(|V|^2) * \mathcal{O}(\textrm{baseline similarity}) \label{eq:complexity_baseline}
\end{align}
Clearly, the complexity of the circle-aware and baseline similarities depends on the specific approach $\phi$ (CN, JC, AA, RA) considered. We focus first on CN. The complexity $ \mathcal{O}(\textrm{similarity CN})$ simply corresponds to an intersection between sets (see line 7 of Algorithm~\ref{algo:sim}), hence it is given by Equation~\ref{eq:complexity_similarity_social} for the social-aware case, and by Equation~\ref{eq:complexity_similarity_baseline} for the baseline. The complexity of a set intersection by using hashing-based dictionaries is $\mathcal{O}(\min(|Set_1|, |Set_2|))$ ~\cite{Philip_2007}.
\begin{align} \label{eq:complexity_similarity_social}
  \mathcal{O}(\textrm{circle-aware similarity}) &= \mathcal{O}(\min(|\Gamma_{C_x}(i)|, |\Gamma_{C_x}(j)|)) \nonumber \\
  & \approx \mathcal{O}(|\Gamma_{C_x}(i)|)) \\
  \mathcal{O}(\textrm{baseline similarity}) &= \mathcal{O}(\min(|\Gamma(i)|, |\Gamma(j)|)) \nonumber \\
  & \approx \mathcal{O}(|\Gamma(i)|)) \label{eq:complexity_similarity_baseline}
\end{align}
While the above formulas have been derived for CN, it is easy to see that also JC, AA, and RA share the same complexity. Indeed, JC leverages the set union in addition to the set intersection, operation that is still linear in the size of the smaller set. AA and RA are the analogous of CN but with different weights, hence, again, they are linear in the size of the smaller set.

We can now substitute Lemma~\ref{theo:time_complexity_egonet} and Equations~\ref{eq:complexity_similarity_social}-\ref{eq:complexity_similarity_baseline} into Equations~\ref{eq:complexity_social} and~\ref{eq:complexity_baseline} above, thus obtaining:
\begin{align*}
  \mathcal{O}(\textrm{circle-aware }\phi) &= \mathcal{O}(|V| * |\Gamma(i)| * \log(|\Gamma(i)|) + \nonumber \\
  & + \mathcal{O}(|V|^2) * |\Gamma_{C_x}(i)| \nonumber{} \\
  & =  \mathcal{O}(|V|^2) * |\Gamma_{C_x}(i)|  \\
  \mathcal{O}(\textrm{baseline }\phi) & = \mathcal{O}(|V|^2) * |\Gamma(i)|.
\end{align*}
The right-hand side of the above equation stems from the fact that social networks are sparse, hence $|\Gamma(i)| \ll |V|$. 
Since the number of neighbors in any circle of the ego network is (significantly) smaller than all neighbors in the baseline (because, by definition, social circles only retain the strongest relationships), we have that the time complexity of the ego-aware system is always (significantly) smaller than the baseline. This concludes the proof.
\end{proof}



In Theorem~\ref{theo:time_complexity_supervised} below, we derive the time complexity for baseline and circle-aware heuristics in the supervised case (discussed in Section~\ref{sec:supervised}). If negatives are undersampled, supervised learning can provide a much greater efficiency than unsupervised link prediction.

\begin{table}[t]
\caption{Time complexity of supervised learning algorithms used in Section~\ref{sec:supervised}. $d$ denotes the number of features considered (in our case $d=4$, i.e., to the number of heuristics defined in Section~\ref{sec:unsupervised}), $k$ is the number of decision trees used in Random Forest (in our case, $k=500$).}
\begin{center}
\resizebox{0.49\textwidth}{!}{%
\begin{tabular}{ l l l l l}
\toprule
 & \textbf{Strategy} &  \textbf{Undersampled negatives} & \textbf{Full dataset} &  \\
\midrule 
\multirow{4}{*}{Training}	& Logistic regr., Na\"ive Bayes  	&  $\mathcal{O}\left(|E_e| d \right)$ 		& $\mathcal{O}\left(|V|^2 d \right) $ \\
					& Decision Trees  	& $\mathcal{O}\left(|E_e| \log(|E_e|) d \right)$ 		& $\mathcal{O}\left(|V|^2 \log(|V|^2) d \right)$\\	
 					& Random Forest	& $\mathcal{O}\left(|E_e| \log(|E_e|) d k \right)$		& $\mathcal{O}\left(|V|^2 \log(|V|^2) d k \right)$\\
 					& SVM			& $\mathcal{O}\left(|E_e|^2\right)$				& $\mathcal{O}\left(|V|^4\right)$\\
\midrule
\multirow{2}{*}{Feature extract.} & Baseline 	& $\mathcal{O}\left(|E_e|\right) \cdot \mathcal{O}\left(|\Gamma(i)|\right)$ 	& $\mathcal{O}\left(|V|^2\right) \cdot \mathcal{O}\left(|\Gamma(i)|\right)$\\
& Circle-aware & $\mathcal{O}\left(|E_e| \right) \cdot \mathcal{O}\left( |\Gamma_{C_x}(i)| \right)$ & $\mathcal{O}\left(|V|^2\right) \cdot \mathcal{O}\left(|\Gamma_{C_x}(i)|\right)$\\
\midrule
\multirow{2}{*}{Egonet comp.} & Baseline 	& - 	& - \\
					      & Circle-aware & \multicolumn{2}{c}{$\mathcal{O}\left(|V| \right) \cdot \mathcal{O}\left(|\Gamma(i)| \cdot \log(|\Gamma(i)|\right))$} \\

\bottomrule
\end{tabular}}
\end{center}
\label{tab:time_complexity_supervised_algos}
\end{table}

\begin{theorem}[Complexity supervised]\label{theo:time_complexity_supervised}
The time complexity of link prediction using circle-aware or baseline heuristics with supervised learning is given by Table~\ref{tab:time_complexity_supervised_algos}. It holds that:
\begin{itemize}
\item When negatives are undersampled, with the exception of logistic regression and na\"ive Bayes, the training phase ($\ge \mathcal{O}(|E_e| \log |E_e|)$) dominates the complexity, and the effect of ego network computation in the feature extraction phase is negligible. 
\item When logistic regression and na\"ive Bayes are used together with negatives undersampling, the overall time complexity is $\mathcal{O}(|E_e|)$, with a slightly smaller multiplicative factor for baseline heuristics.
\item When negatives are not undersampled, the computation of ego networks (which is $\mathcal{O}(|V|) \approx \mathcal{O}(|E_e|)$) is negligible, and the feature extraction is slightly more convenient for circle-aware heuristics. 
\end{itemize}
\end{theorem}

\begin{proof}
The complexity of supervised learning can be decomposed into the complexity for extracting the relevant features and the complexity of solving the learning problem on these features:
\begin{multline*}
\mathcal{O}(\textrm{supervised link prediction}) = \\
\mathcal{O}(\textrm{feature extraction}) + \mathcal{O}(\textrm{supervised learning})
\end{multline*}
It is easy to see that the latter component is the same regardless of the social circle considered, since the features described in Section~\ref{sec:unsupervised}, once computed, are simply the four heuristics associated with each edge $(i,j)$. Specifically, denoting the number of features with $d$ (with $d=4$ in our case) and the number of training links with $|E_{train}|$, the time complexity for training the supervised algorithms used in Section~\ref{sec:supervised} is: $\mathcal{O}\left(d |E_{train}|\right)$ for logistic regression and na\"ive Bayes, $\mathcal{O}\left(|E_{train}| \log(|E_{train}|) d \right)$ for decision trees, $\mathcal{O}\left(|E_{train}| \log(|E_{train}|) d k \right)$ for Random Forest ($k$ is the number of decision trees used in Random Forest), $\mathcal{O}\left(|E_{train}|^2\right)$ for SVM (of the latter, we here consider the kernel version).
%
%
%
The number $|E_{train}|$ of training links depends on whether we consider the undersampled or full training set. In the former case, $|E_{train}| = 2 |E_e|$, since the negative training links are undersampled to match the number of positive ones. In the latter, $|E_{train}| \sim \mathcal{O}\left( | V^2 | \right)$, since social networks are sparse. Note that the time complexity of the learning phase only depends on the selected supervised learning algorithm.

We now focus on $\mathcal{O}(\textrm{feature extraction})$. 
We start with the circle-aware feature extraction. In the equation below, we derive the time complexity separating the negatives undersampling and the full dataset case, and leveraging Lemma~\ref{theo:time_complexity_egonet} for $\mathcal{O}(\textrm{ego})$ and Equation~\ref{eq:complexity_similarity_social} for $\mathcal{O}(\textrm{circle-aware sim})$. In general, the similarity is computed for all node pairs in $E_{train} = E_e \cup \overline{E}_e^{train}$, where $\overline{E}_e^{train}$ denotes the set of negative training edges. However, with negatives undersampling, $|\overline{E}_e^{train}| = |E_e|$, with the full dataset $\overline{E}_e^{train} \approx \mathcal{O}(|V|^2)$.
\begin{multline}\label{eq:complexity_feature_extr_circles}
\mathcal{O}\left(\textrm{circle-aware feat. extr.} \right) = \\
= \mathcal{O}(|V|) \cdot \mathcal{O}(\textrm{ego}) + \bigg\{ \begin{matrix} \mathcal{O}(|E_e|) \cdot \mathcal{O}(\textrm{circle-aware sim}) \\ \mathcal{O}(|V|^2) \cdot \mathcal{O}(\textrm{circle-aware sim})\end{matrix} \\
\approx \bigg\{ \begin{matrix} \mathcal{O}(|E_e|) \cdot \mathcal{O}(|\Gamma(i)| \cdot \log(|\Gamma(i)|) \cdot |\Gamma_{C_x}(i)|) \\ \mathcal{O}(|V|^2) \cdot \mathcal{O}(|\Gamma_{C_x}(i)|)\end{matrix}
\end{multline}
Deriving the complexity of baseline feature extraction is now straightforward. In fact, we can simply neglect the computation of ego networks in Equation~\ref{eq:complexity_feature_extr_circles} and use the formula for  $\mathcal{O}(\textrm{baseline sim})$ in Equation~\ref{eq:complexity_similarity_baseline}. Then, the equation below follows:
\begin{multline}\label{eq:complexity_feature_extr_baseline}
\mathcal{O}\left(\textrm{baseline feature extraction} \right) = \\
	= \bigg\{ \begin{matrix} \mathcal{O}(|E_e|) \cdot \mathcal{O}(\textrm{baseline sim}) \\ \mathcal{O}(|V|^2) \cdot \mathcal{O}(\textrm{baseline sim})\end{matrix} = \\
	= \bigg\{ \begin{matrix} \mathcal{O}(|E_e|) \cdot \mathcal{O}(|\Gamma(i)|) \\ \mathcal{O}(|V|^2) \cdot \mathcal{O}(|\Gamma(i)|)\end{matrix}
\end{multline}
Comparing Equations~\ref{eq:complexity_feature_extr_circles} and~\ref{eq:complexity_feature_extr_baseline}, we remark that, when using the full dataset, the time complexity of feature extraction is dominated by $\mathcal{O}(|V|^2)$, with circle-aware features being slightly more convenient (since $|\Gamma_{C_x}(i)| < |\Gamma(i)|$). Vice versa, with negatives undersampling, the time complexity of feature extraction is $\mathcal{O}(|E_e|)$ and baseline features are marginally more advantageous.
\end{proof}

If we compare the heuristic-based link prediction complexity in Theorems~\ref{theo:time_complexity_unsupervised}-\ref{theo:time_complexity_supervised}, we observe that supervised link prediction with negatives undersampled is more efficient than unsupervised link prediction, as it allows to reduce the complexity from $\mathcal{O}(|V|^2)$ to $\mathcal{O}(|E_e|)$. In these settings, the small overhead of computing ego networks is paid off by the outstanding prediction performance. We can now derive the time complexity of \emph{node2vec}-based link prediction and SEAL. 

\begin{theorem}
The time complexity of link prediction with \emph{node2vec}, using Hadamard product for edge embedding, is given by the following:
\begin{itemize}
\item  $\mathcal{O}(d |E_e|)$ when negatives are undersampled in the training set,
\item  $\mathcal{O}(d |V|^2)$ when all training negatives are used.
\end{itemize}

\end{theorem}

\begin{proof}
The time complexity for getting the node embedding with \emph{node2vec} is $\mathcal{O}(d |V|)$ \cite{zhang2018network-representation}, where $d=128$, i.e., the dimensions of the embeddings. . Then, for each training edge, we have to extract the edge embedding by taking the Hadamard product of the node embedding pairs for training links:
\begin{equation*}
\mathcal{O}\left(d \cdot |E_{train}| \right) \approx \bigg\{ \begin{matrix} \mathcal{O}(d |E_e|) & \textrm{with negatives undersampling}  \\ \mathcal{O}(d |V|^2) & \textrm{with full training data.} \\\end{matrix}
\end{equation*}
Finally, for link prediction, we train a logistic regression model, whose complexity is $\mathcal{O}\left(d \cdot |E_{train}| \right)$, hence the above equation also holds for the training phase. Thus, the thesis follows.
\end{proof}

\begin{theorem}
The time complexity of link prediction using SEAL is:
\begin{itemize}
\item  $\mathcal{O}(d |E_e| |\Gamma(i)|^h|)$ when negatives are undersampled in the training set,
\item  $\mathcal{O}(|V|^2 |\Gamma(i)|^h)$ when all training negatives are used.
\end{itemize}
\end{theorem}

\begin{proof}
SEAL starts with extracting enclosing subgraphs of depth $h$ for each training link. We denote the edge set of the enclosing subgraphs of depth $h$ as $E_{sub}^h$. Note that $|E_{sub}^h| \approx 2 |\Gamma(i)|^h$, and $|E_{sub}^h|$ approaches $|E|$ as $h$ gets larger. This phase requires $\mathcal{O}\left( |E_{train}| \cdot |E_{sub}^h| \right)$. Then, each enclosing subgraph is fed to DGCNN, whose complexity is linear in the number of edges of the training graph~\cite{wu2020comprehensive}. Since DGCNN is applied to each enclosing subgraph, the complexity of this phase is given by $\mathcal{O}\left( |E_{train}| \cdot |E_{sub}^h| \right)$.
When SEAL is used together with \emph{node2vec} embeddings, we have to also consider an initial $\mathcal{O}(d |V|)$ (with $d=128$) for computing the embeddings, which can also be approximated as $\mathcal{O}(d |E_e|)$. 
\end{proof}

We conclude this section by comparing the complexity of the best-performing link prediction algorithms according to the results of Section~\ref{sec:evaluation}: circle-aware supervised link prediction (with logistic regression or na\"ive Bayes), \emph{node2vec}, and SEAL.  Considering the case where negatives are undersampled (more effective for all strategies, including self-supervised ones), all the best-performing algorithms are linear in the number of edges. Hence, our circle-aware approach is able to outperform state of the art solutions without worsening the computational complexity.

}{
}

\end{document}